\pdfoutput=1

\documentclass[a4paper,UKenglish,cleveref, autoref, thm-restate]{lipics-v2019}
\usepackage{graphicx}
\usepackage{tikz}
\usepackage{ifthen}
\usetikzlibrary{calc}
\usepackage{subcaption}
\usepackage{caption}

\newcommand{\calR}{\mathcal{R}}

\newcommand{\alp}{A}
\newcommand{\conf}[2]{\textup{C}(#1,#2)}

\newcommand{\reach}{\textsc{Reach}}
\newcommand{\average}{\textsc{Avg}}

\newcommand{\underpopulation}[2]{\mathcal{R}^{+}(#1,#2)}
\newcommand{\overpopulation}[2]{\mathcal{R}^{-}(#1,#2)}

\newcommand{\word}[1]{\tau_{#1}}
\newcommand{\wildcard}{?}

\newcommand{\numbu}[1]{[ #1 ]}

\newcommand*\calO{\mathcal{O}}

\title{Simplified Game of Life: Algorithms and Complexity}

\author{Krishnendu Chatterjee}
  {Institute of Science and Technology, Austria}
  {}{}
  {The research was partially supported by the Vienna Science and Technology Fund (WWTF) Project ICT15-003.}
 
\author{Rasmus Ibsen-Jensen}
  {University of Liverpool, United Kingdom}{}{}{}{}

\author{Isma\"el Jecker}
  {Institute of Science and Technology, Austria}
 {}{}
 {This project has received funding from the European Union’s Horizon 2020 research and innovation programme under the Marie Skłodowska-Curie Grant Agreement No. 754411.}{}{}

\author{Jakub Svoboda}
  {Institute of Science and Technology, Austria}
  {}{}{}{}
\authorrunning{K. Chatterjee, R. Ibsen-Jensen, I. Jecker and J. Svoboda}

\Copyright{Krishnendu Chatterjee, Rasmus Ibsen-Jensen, Ismael Jecker, Jakub Svoboda} 



\nolinenumbers %uncomment to disable line numbering

\hideLIPIcs  %uncomment to remove references to LIPIcs series (logo, DOI, ...), e.g. when preparing a pre-final version to be uploaded to arXiv or another public repository

\EventEditors{Javier Esparza and Daniel Kr{\'a}l'}
\EventNoEds{2}
\EventLongTitle{45th International Symposium on Mathematical Foundations of Computer Science (MFCS 2020)}
\EventShortTitle{MFCS 2020}
\EventAcronym{MFCS}
\EventYear{2020}
\EventDate{August 24--28, 2020}
\EventLocation{Prague, Czech Republic}
\EventLogo{}
\SeriesVolume{170}
\ArticleNo{16}

\begin{document}
\maketitle
\begin{abstract}
Game of Life is a simple and elegant model to study dynamical system over
networks.
The model consists of a graph where every vertex has one of two types, namely,
dead or alive.
A configuration is a mapping of the vertices to the types.
An update rule describes how the type of a vertex is updated given the types of
its neighbors.
In every round,
all vertices are updated synchronously, which leads to a configuration
update.
While in general, Game of Life allows a broad range of update rules,
we focus on two simple families of update rules,
namely, underpopulation and overpopulation,
that model several interesting dynamics studied in the literature.
In both settings, a dead vertex requires at least a 
desired number of live neighbors to become alive.
For underpopulation (resp., overpopulation),
a live vertex requires at least (resp. at most) 
a desired number of live neighbors to remain alive.
We study the basic computation problems, e.g., configuration reachability, for these 
two families of rules.
For underpopulation rules,
we show that these problems can be solved
in polynomial time,
whereas for overpopulation rules they are PSPACE-complete.
\end{abstract}

\section{Introduction}

Game of Life is a well-known model to study dynamics over networks.
We consider the classical model of Game of Life and study two simple
update rules for which we establish algorithms and computational complexity.
We start with a description of dynamical systems, then explain Game of Life and our simplified rules,
and finally state our main results.

\subparagraph{Dynamical systems.}
A dynamical system describes a set of rules updating the state of the system. 
The study of dynamical systems and computational
questions related to them is a core problem in computer science. 
Some classic examples of dynamical systems are the following:
(a)~a set of matrices that determine the update of the state of the dynamical system~\cite{blondel2001stability,ouaknine2019decidability};
(b)~a stochastic transition matrix that determines the state update (the classical model of
Markov chains)~\cite{Kemeny}; 
(c)~dynamical systems that allows stochastic and non-deterministic behavior (aka Markov 
decision processes)~\cite{FV97,Puterman,BaierBook};
and so on.
To study dynamics over networks the two classical models are
Game of Life~\cite{gardener1970mathematical,berlekamp2004winning} and
cellular automata~\cite{wolfram2018cellular}.

\subparagraph{Game of Life.}
Game of Life is a simple yet elegant model to study dynamical systems over networks. 
The network is modeled as a graph where every vertex has one of two types, namely, dead or alive. 
A configuration (or state of the system) is a mapping of the vertices to the types. 
An update rule describes how the type of a vertex is updated given the types of
its neighbors.
In every round, all vertices are updated synchronously, leading to a configuration
update.
In Game of Life, the update rules are deterministic, hence the configuration graph is
deterministic.
In other words, from each starting configuration, the updates lead to a finite path 
followed by a cycle in the configuration graph.
While in Game of Life the successor of a state only depends
on its number of neighbors of each type, in the more general
model of cellular automata the update rule can also
distinguish among the positions
of the neighbors.

\subparagraph{Simplified rules.} 
While the update rules in Game of Life are quite general,
in this work, we focus on two simplified rules,
namely, {\em underpopulation rules} and {\em overpopulation rules}:
\begin{itemize}
\setlength{\itemsep}{0pt}
\setlength{\parskip}{0pt}
\setlength{\parsep}{0pt}
\item {\em Underpopulation rule.} 
According to an underpopulation rule, 
a dead vertex becomes alive if it
has at least $i_0$ live neighbors,
and remains dead otherwise;
a live vertex remains alive if
it has at least $i_1$ live neighbors,
and becomes dead otherwise.
 \item {\em Overpopulation rule.} 
According to an overpopulation rule, 
a dead vertex becomes alive if it
has at least $i_0$ live neighbors,
and remains dead otherwise;
a live vertex remains alive if
it has at most $i_1$ live neighbors,
and becomes dead otherwise.
\end{itemize}
\noindent
See Section~\ref{section:Preliminaries} for the formal details of the definition, and a
detailed comparison of our setting with cellular automata,
and Conway's original Game of Life.

\subparagraph{Motivation.}
While we consider simpler rules, we study their effect on
any type of graph, contrary to cellular automata that focus on grids.
This allows us to model several complex situations studied in the literature.
For example, the underpopulation rule models the spread of ideas,
where a person adopts a new idea only if sufficiently many neighbors adopt it, or study bandwagon effects where
an item is considered useful if sufficiently many neighbors use it.
In contrast, the overpopulation rule models anti-coordination effects where the goal is to avoid
a popular pub, or model snob effects where individuals discard a fashion if too many neighbors have adopted it. 
See Section~\ref{section:motivation} for further details.

\subparagraph{Basic computational problems.}
We study two basic computational problems for the underpopulation and overpopulation rule.
The first computational problem is the {\em configuration reachability} question which asks,
given an initial configuration and a target configuration,
whether the target configuration is reachable from the initial configuration.
The second computational problem is the {\em long-run average} question,
which asks, given an initial configuration,
what is the long-run average of the number of live vertices.
Note that in the configuration graph,
any initial configuration is the source of a finite path followed by a cycle.
The long-run average question asks about the average number of live vertices in the cycle.

\subparagraph{Our contributions.}
Our main contributions are algorithms and complexity results for the two families of rules.
First, for the underpopulation rules,
we present polynomial time algorithms
for both computational problems for all graphs. 
Thus, we identify a simple update rule in Game of Life,
that can model several interesting scenarios, 
for which we present efficient algorithms.
Second, for the overpopulation rules,
we show that both computational problems are PSPACE-complete. 
Note that the PSPACE upper bound holds for general update rules for Game of Life,
hence the main contribution is the PSPACE hardness proof. 
Moreover, we show that the PSPACE hardness even holds for regular graphs with a constant degree.
Note that the difference between underpopulation and overpopulation is minimal
(one inequality reversed), 
yet we show that while efficient algorithms exist for underpopulation rules,
the computational problems for overpopulation rules are intractable.

\section{Preliminaries}\label{section:Preliminaries}

Given a finite alphabet $\alp$,
we denote by $\alp^*$
the set of finite sequences
of elements of $\alp$,
and by $\alp^{\omega}$
the set of infinite sequences
of elements of $\alp$.
The elements of $\alp^*$ and $\alp^{\omega}$
are called \emph{words} over $\alp$.
The \emph{length}
of a word $w = a_1a_2a_3 \ldots \in A^* \cup A^\omega$
is its number of letters,
denoted $|w| \in \mathbb{N} \cup \{ + \infty \}$.
A \emph{factor} of $w$ is a sequence of consecutive letters of $w$.
For every $0 \leq i \leq j \leq |w|$,
we denote by $w[i,j]$ the factor
$a_{i+1}a_{i+2} \ldots a_j$
of $w$.

A (finite) \emph{graph} is a pair $G = (V,E)$
composed of a finite set of \emph{vertices} $V$
and a set of \emph{edges} $E \subseteq V \times V$
that are pairs of vertices.
A \emph{walk} of $G$ is a sequence
$\rho=s_1,s_2,s_3 \ldots \in V^* \cup V^\omega$
such that each pair of consecutive 
vertices is an edge:
$(s_i,s_{i+1}) \in E$
for every $1 \leq i < |\rho|$.
A (simple) \emph{path} is a walk whose
vertices are all distinct.
A (simple) \emph{cycle} is a walk in which
the first and last vertices are identical,
and all the other vertices are distinct.
A graph is called \emph{undirected}
if its set of edges is symmetric:
$(s,t) \in E \Leftrightarrow (t,s) \in E$.
Two vertices of an undirected graph are called
\emph{neighbors} if they are linked by an edge.

\subsection{Configurations and update rules}

A \emph{configuration} of a graph
is a mapping of its vertices
into the set of states $\{0,1\}$.
We say that a vertex is \emph{dead}
if it is in state $0$, and \emph{alive}
if it is in state $1$.
An \emph{update rule} $\calR$
is a set of
\emph{deterministic},
\emph{time-independant},
and \emph{local}
constraints
determining the evolution of configurations of a graph:
the successor state of each vertex is determined
by its current state and the states of its neighbors.
We define an update rule formally as a pair of functions:
for each $q \in \{0,1\}$,
the \emph{state update function} $\phi_q$
maps any possible neighborhood configuration
to a state in $\{0,1\}$.
The successor state of a vertex
in state $q$ with neighborhood in configuration $c_n$
is then defined as $\phi_q(c_n) \in \{0,1\}$.

In this work, we study
the effect on undirected graphs
of update rules 
definable by state update functions
that are \emph{symmetric}
and \emph{monotonic}
(configurations are partially ordered
by comparing through inclusion
their subsets of live vertices).
In this setting,
a vertex is not able to differentiate its neighbors,
and has to determine its successor state
by comparing the number of its live neighbors with a threshold.
These restrictions give rise to four families of rules,
depending on whether $\phi_0$ and $\phi_1$
are monotonic increasing or decreasing.
We study the two families corresponding to increasing $\phi_0$,
the two others can be dealt with by using symmetric arguments.

\subparagraph{Underpopulation.}
An \emph{underpopulation (update) rule} $\underpopulation{i_0}{i_1}$
is defined by two thresholds:
$i_0 \in \mathbb{N}$ is the minimal number of live neighbors
needed for the birth of a dead vertex,
and $i_1 \in \mathbb{N}$
is the minimal number of live neighbors
needed for a live vertex to stay alive.
Formally, the successor $\phi_q(m)$
of a vertex currently in state $q \in \{0,1\}$
with $m \in \mathbb{N}$ live neighbors is
\[
\begin{array}{ll}
\phi_0(m) = \left\{
\begin{array}{lll}
0 \textup{ if } m < i_0;\\
1 \textup{ if } m \geq i_0.
\end{array}
\right. &
\phi_1(m) = \left\{
\begin{array}{lll}
0 \textup{ if } m < i_1;\\
1 \textup{ if } m \geq i_1.
\end{array}
\right.
\end{array}
\]
This update rule is symmetric and monotonic.

\subparagraph{Overpopulation.}
An \emph{overpopulation (update) rule} $\overpopulation{i_0}{i_1}$
is defined by two thresholds:
$i_0 \in \mathbb{N}$
is the minimal number of live neighbors
needed for the birth of a dead vertex,
and $i_1 \in \mathbb{N}$
is the maximal number of live neighbors
allowing a live vertex to stay alive.
Formally, the successor $\phi_q(m)$
of a vertex currently in state $q \in \{0,1\}$
with $m \in \mathbb{N}$ live neighbors is
\[
\begin{array}{ll}
\phi_0(m) = \left\{
\begin{array}{lll}
0 \textup{ if } m < i_0;\\
1 \textup{ if } m \geq i_0.
\end{array}
\right. &
\phi_1(m) = \left\{
\begin{array}{lll}
0 \textup{ if } m > i_1;\\
1 \textup{ if } m \leq i_1.
\end{array}
\right.
\end{array}
\]
This update rule is symmetric and monotonic.

\subparagraph{Basic computational problems.}
To gauge the complexity
of an update rule,
we study two corresponding computational problems.
Formally, given an update rule $\calR$ and a graph $G$,
the \emph{configuration graph} $\conf{G}{\calR}$
is the (directed) graph whose vertices are the
configurations of $G$,
and whose edges are the pairs
$(c,c')$
such that the configuration $c'$ is successor of $c$
according to the update rule $\calR$.
Note that $\conf{G}{\calR}$ is finite since $G$ is finite.
Moreover,
since the update rule $\calR$ is deterministic,
every vertex of the configuration graph
is the source of a single
infinite walk composed of a finite path followed by a cycle.
\begin{itemize}
\setlength{\itemsep}{0pt}
\setlength{\parskip}{0pt}
\setlength{\parsep}{0pt}
\item
The \emph{configuration reachability problem},
denoted \reach{},
asks, given a graph $G$,
an initial configuration $c_{I}$,
and a final configuration $c_{F}$,
whether the walk in $\conf{G}{\calR}$
starting from $c_{I}$
eventually visits $c_{F}$.
\item
The \emph{long-run average problem},
denoted \average{},
asks, given a threshold $\delta \in [0,1]$,
a graph $G$, and an initial configuration $c_I$,
whether $\delta$ is strictly smaller than
the average ratio of live vertices
in the configurations that are part of the
cycle in $\conf{G}{\calR}$ reached from $c_I$.
\end{itemize}

\subsection{Comparison to other models}
We show similarities and differences
between our update rules and similar models.

\subparagraph{Cellular automata.}
Cellular automata study update rules defined on
(usually infinite) grid graphs~\cite{wolfram2018cellular}.
Compared to the setting studied in this paper,
more rules are allowed
since neither symmetry nor monotonicity is required,
yet underpopulation and overpopulation rules
are not subcases of cellular automata, 
as they are defined for any type of graph,
not only grids.
To provide an easy comparison
between the update rules studied in this paper
and some well-studied
cellular automata,
we now define Rule $54$ and Rule $110$
(according to the numbering scheme presented in~\cite{wolfram1983statistical})
using the formalism of this paper.

\smallskip
\noindent
\textbf{1.} Rule $54$~\cite{boccara1991particlelike,martinez2006phenomenology} coincides with
the overpopulation rule $\overpopulation{1}{0}$ 
applied to the infinite unidimensional linear graph.
A dead vertex becomes alive if
one of its neighbors is alive,
and a live vertex stays alive
only if both its neighbors are dead.
Formally, the successor $\phi_q(m)$
of a vertex currently in state $q \in \{0,1\}$
with $m \in \{0,1,2\}$ live neighbors is
\[
\begin{array}{ll}
\phi_0(m) = \left\{
\begin{array}{lll}
0 \textup{ if } m = 0;\\
1 \textup{ if } m \geq 1.
\end{array}
\right. &
\phi_1(m) = \left\{
\begin{array}{lll}
0 \textup{ if } m \geq 1;\\
1 \textup{ if } m = 0.
\end{array}
\right.
\end{array}
\]
This update rule is symmetric and monotonic.
It can be used to model logical gates~\cite{martinez2006phenomenology},
and is conjectured to be Turing complete.

\smallskip
\noindent
\textbf{2.} Rule $110$~\cite{cook2004universality} is defined over the infinite unidimensional linear graph.
A dead vertex copies the state of its right neighbor,
and a live vertex stays alive
as long as at least one of its neighbors is alive.
Formally, the successor $\phi_q(\ell,r)$
of a vertex currently in state $q \in \{0,1\}$
with left neighbor in state $\ell \in \{0,1\}$
and right neighbor in state $r \in \{0,1\}$ is
\[
\begin{array}{ll}
\phi_0(\ell,r) = r;&
\phi_1(\ell,r) = \left\{
\begin{array}{lll}
0 \textup{ if } \ell = r = 0;\\
1 \textup{ otherwise}.
\end{array}
\right.
\end{array}
\]
This update rule is monotonic, but not symmetric.
It is known to be Turing complete.

\subparagraph{Game of Life.}
Game of Life requires update rules that are symmetric,
but not necessarily monotonic.
The most well known example is
Conway's Game of Life~\cite{gardener1970mathematical,berlekamp2004winning},
that has been adapted in various ways,
for example as 3-D Life~\cite{bays1987candidates}, or the beehive rule~\cite{wuensche2004self}.
Conway’s game of life studies the evolution of the infinite two-dimensional square
grid according to the following update rule: a dead vertex
becomes alive if it has exactly three live neighbors, and a
live vertex stays alive if it has two or three live neighbors.
Formally, the successor $\phi_q(m)$
of a vertex currently in state $q \in \{0,1\}$
with $m \in \mathbb{N}$ live neighbors is
\[
\begin{array}{ll}
\phi_0(m) = \left\{
\begin{array}{lll}
0 \textup{ if } m \neq 3;\\
1 \textup{ if } m = 3.
\end{array}
\right. &
\phi_1(m) = \left\{
\begin{array}{lll}
0 \textup{ if } m \notin \{2,3\};\\
1 \textup{ if } m \in \{2,3\}.
\end{array}
\right.
\end{array}
\]
This update rule is symmetric,
but not monotonic.
It is known to be Turing complete \cite{berlekamp2004winning}.
\section{Motivating Examples}\label{section:motivation}

Our dynamics can represent situations where an individual (a
vertex) adopts a behavior (or a strategy) only if the behavior
is shared by sufficiently many neighboring individuals. Then,
the underpopulation setting corresponds to behaviors that
are abandoned if not enough neighbors keep on using it,
while the overpopulation setting models behaviors that are
dropped once they are adopted by too many.
We present several examples.

\subsection{Underpopulation}

\subparagraph{Innovation.}
The problem of spreading innovation is considered in~\cite{innovation_book, innovation_diffusion}.
Initially, a small group of people starts using a new product and if others see it
used, they adopt the innovation.
In our setting this corresponds to the underpopulation model.
The question of determining whether the innovation gets to some key people can be formalised
as the configuration reachability problem \reach{},
and predicting how many people will eventually be using the innovation
amounts to solve the long-run average problem \average{}.
Similar questions are asked in~\cite{public_opinion_formation},
where the authors study how opinions form.
See Appendix~\ref{appendix:motivation} for more details.

\subparagraph{Positive feedback.}
In the paper~\cite{bandwagon_snob_effects},
the bandwagon and Veblen effects are described.
These consider a fact that the demand for an item
increases with the number of people using that item.
Under this hypothesis, determining the demand
corresponds to solve \average{} for an underpopulation rule.
Many more examples, for example, how people behave depending on what their friends do, can be found
in~\cite{connected_book}.
Anything from emotions to obesity can spread through a network,
usually following small modifications of the underpopulation rule.

\subsection{Overpopulation}

\subparagraph{Anticoordination.}
Imagine that you want to go mushroom hunting. You enjoy the peaceful walk in the forest and
love the taste of fried wild mushrooms, or mushroom soup.
Mushrooming is a solitary activity and if too many of your neighbors decide to
go mushrooming too,
they annoy you, and you find fewer mushrooms in already searched forest.
So, if you were not mushrooming the day before you can get convinced that the mushrooms are growing
by some neighbors that show you baskets full of delicious mushrooms.
However, if you decide to go and see too many people there,
you get discouraged and do not go the next day.

This behavior is called anticoordination and was described in~\cite{bounded_rationality},
it more generally describes optimal exploitation of resources.
The questions here are: does some set of people go mushroom hunting, how many people will be mushroom hunting?
The overpopulation closely corresponds to this with \reach{} and \average{} answering the questions.

\subparagraph{Snob effect.}
Many items are desirable because they are expensive, or unique.
This behavior was observed in~\cite{bandwagon_snob_effects}.
People start doing something, but if too many people do it, it loses appeal.
For instance fashion works this way for all of us: People get inspired by what
they see, but if too many people wear the same outfit, they change it.

\section{Underpopulation: PTIME Algorithm}

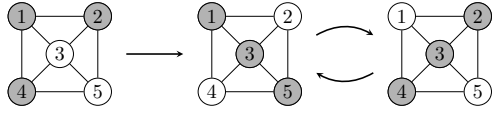
\begin{figure}
  \begin{center}
    \scalebox{0.5}{
    \begin{tikzpicture}
      \foreach \c in {1,6,11}{
        \pgfmathtruncatemacro{\ste}{\c}
        \pgfmathtruncatemacro{\ma}{0}\
        \coordinate (c\c1) at (\ste,\ma+2);
        \node[draw,circle,fill=white,scale=1.5,inner sep=2]
          (n\c1) at (c\c1) {4};
        \coordinate (c\c2) at (\ste,\ma+4);
        \node[draw,circle,fill=white,scale=1.5,inner sep=2]
          (n\c2) at (c\c2) {1};
      }
      \foreach \c in {3,8,13}{
        \pgfmathtruncatemacro{\ste}{\c}
        \pgfmathtruncatemacro{\ma}{0}\
        \coordinate (c\c1) at (\ste,\ma+2);
        \node[draw,circle,fill=white,scale=1.5,inner sep=2]
          (n\c1) at (c\c1) {5};
        \coordinate (c\c2) at (\ste,\ma+4);
        \node[draw,circle,fill=white,scale=1.5,inner sep=2]
          (n\c2) at (c\c2) {2};
      }
      \foreach \c in {2,7,12}{
        \pgfmathtruncatemacro{\ste}{\c}
        \pgfmathtruncatemacro{\ma}{3}
        \pgfmathtruncatemacro{\height}{(\ma)}
        \pgfmathtruncatemacro{\ste}{\c}
        \coordinate (c\c1) at (\ste,\height);
        \node[draw,circle,fill=white,scale=1.5,inner sep=2]
          (n\c1) at (c\c1) {3};
      }
      \foreach \c in {1,6,11}{
        \pgfmathtruncatemacro{\nx}{\c+1}
        \pgfmathtruncatemacro{\nnx}{\c+2}
          \draw (n\c1) -- (n\c2);
          \draw (n\c1) -- (n\nx1);
          \draw (n\c2) -- (n\nx1);
          \draw (n\c1) -- (n\nnx1);
          \draw (n\c2) -- (n\nnx2);
          \draw (n\nx1) -- (n\nnx1);
          \draw (n\nx1) -- (n\nnx2);
          \draw (n\nnx1) -- (n\nnx2);
         
      }
      \node[draw,circle,fill=black!30,scale=1.5,inner sep=2]
       (n11) at (c11) {4};
      \node[draw,circle,fill=black!30,scale=1.5,inner sep=2]
       (n12) at (c12) {1};
      \node[draw,circle,fill=black!30,scale=1.5,inner sep=2]
       (n32) at (c32) {2};
      \node[draw,circle,fill=black!30,scale=1.5,inner sep=2]
       (n62) at (c62) {1};
      \node[draw,circle,fill=black!30,scale=1.5,inner sep=2]
       (n71) at (c71) {3};
      \node[draw,circle,fill=black!30,scale=1.5,inner sep=2]
       (n81) at (c81) {5};
      \node[draw,circle,fill=black!30,scale=1.5,inner sep=2]
       (n111) at (c111) {4};
      \node[draw,circle,fill=black!30,scale=1.5,inner sep=2]
       (n121) at (c121) {3};
      \node[draw,circle,fill=black!30,scale=1.5,inner sep=2]
       (n132) at (c132) {2};

      \draw[-{stealth},very thick]
       ($(c21) + (1.75,0)$) to ($(c71)-(1.75,0)$);
      \draw[-{stealth},very thick]
       ($(c71) + (1.75,0.5)$) to [bend left] ($(c121)-(1.75,-0.5)$);
      \draw[-{stealth},very thick]
       ($(c121) - (1.75,0.5)$) to [bend left] ($(c71)+(1.75,-0.5)$);
    \end{tikzpicture}
    }
  \end{center}
  \caption{Evolution of a graph
  under the underpopulation rule
  $\underpopulation{2}{2}$.
  Live vertices are gray.}\label{underpopulation_example_picture}
\end{figure}

In this section, we study underpopulation rules:
a vertex comes to life if it has
sufficiently many living neighbours,
and then still requires enough
living neighbours to stay alive.
Our main result is as follows:

\begin{restatable}{theorem}{theoremDecidable}\label{theorem:underpopulation}
For every underpopulation rule,
the reachability and long-run average problems
are decidable in polynomial time.
\end{restatable} 

The above result depends on two key propositions.
Let us start by having a look at an example.
Figure \ref{underpopulation_example_picture}
presents successive configurations of a graph
where each vertex requires at least
$2$ living neighbours
in order to be alive in the next step.
The resulting behaviour is quite simple:
after the initial step,
the graph keeps on oscillating between
two configurations:
the middle vertex has reached a stable state,
and the other vertices alternate between being dead and alive.
We show that the behaviour of graphs 
following underpopulation rules
can actually never be much more complicated than this.
First, no huge cycle of configurations can happen.

\begin{restatable}{proposition}{theoremCycle}\label{theorem:underpopulation_cycle}
For every $i_0,i_1 \in \mathbb{N}$ and
every undirected graph $G$,
the configuration graph
$\conf{G}{\underpopulation{i_0}{i_1}}$
admits no simple cycle of length bigger than two.
\end{restatable}

Moreover, a cycle is always reached early in the process.

\begin{restatable}{proposition}{propositionReach}\label{proposition:underpopulation}
For every $i_0,i_1 \in \mathbb{N}$ and
every undirected graph $G$,
the configuration graph
$\conf{G}{\underpopulation{i_0}{i_1}}$
admits no simple path of length
$2|E| + 2(i_0+i_1+1)|V| + 4$ or more.
\end{restatable}

We now present the proof of Theorem \ref{theorem:underpopulation},
by showing that Proposition
\ref{proposition:underpopulation}
yields a polynomial time algorithm
for both \reach{}
and \average{}.

\begin{proof}[Proof of Theorem \ref{theorem:underpopulation}]
Since underpopulation rules are deterministic,
once a vertex is repeated in a walk
in $\conf{G}{\underpopulation{i_0}{i_1}}$,
no new configuration can be visited.
Therefore,
Proposition \ref{proposition:underpopulation}
bounds polynomially the number of configurations reachable from
an initial configuration.
Since computing the successor of a configuration
and checking the equality of configurations
can both be done in polynomial time,
we obtain the following polynomial time
algorithms solving \reach{}
and \average{}:
first, we list all the configurations reachable
from the initial configuration,
then we check if the final configuration
is part of it,
respectively if the rate of live vertices
in the reached loop
is higher than the required threshold.
\end{proof}

The remainder of this section is
devoted to the proof of 
Propositions \ref{theorem:underpopulation_cycle}
and 
\ref{proposition:underpopulation}.
Let us fix an underpopulation rule
$\underpopulation{i_0}{i_1}$,
a graph $G = (V,E)$,
and an initial configuration of $G$.
Our proofs rely on a key lemma that sets a bound
on the number of times
a vertex of $G$ switches its state between two
configurations separated by two time steps.
We begin by introducing some technical
concepts and notations,
then we state our key lemma (Subsection \ref{subsec_under_key}).
Afterwards, we proceed with the formal proofs
of Propositions \ref{theorem:underpopulation_cycle} and \ref{proposition:underpopulation}
(Subsection \ref{subsec_under_proof}).

\subsection{Key Lemma}\label{subsec_under_key}

We begin by introducing some auxiliary notation,
and then we state our key lemma.

\subparagraph{Histories.}
The \emph{history} of a vertex $s \in V$ is the infinite word 
$\word{s} \in \{0,1\}^\omega$
whose letters are the successive states of $s$.
For instance, in the setting 
depicted in Figure \ref{underpopulation_example_picture},
the histories are:
\[
\tau_1 = 1(10)^{\omega} \ \ \
\tau_2 = (10)^{\omega} \ \ \
\tau_3 = 01^{\omega} \ \ \
\tau_4 = (10)^{\omega} \ \ \
\tau_5 = (01)^{\omega}
\]
The state of vertex $3$ stabilises after the first step,
and the other four vertices end up oscillating between two states.
The proofs of this section rely on counting, in the histories of $G$,
the number of factors (sequence of consecutive letters)
matching some basic regular expressions.
In order to easily refer to these numbers,
we introduce the following notations.
Let us consider the alphabet $\{0,1,\wildcard \}$,
where $\wildcard$ is a wildcard symbol
matching both $0$ and $1$.
Given an integer
$m \in \mathbb{N}$
and a word $y \in \{0,1,\wildcard \}^*$
of size $n \leq m$,
we denote by $\numbu{y}_m \in \mathbb{N}$
the number of factors of the prefixes $\word{s}[0,m]$, $s \in V$,
that match the expression $y$.
Formally,
\begin{align*} 
&\numbu{y}_m
= \big{|}
\{
(s,i) \in V \times \mathbb{N} |
i+n \leq m,
\word{s}[i,i+n] = y
\}
\big{|}.
\end{align*}
Additional definitions that are required for the technical proofs of this paper,
along with examples illustrating them, can be found in the appendix (Part \ref{appendix:examples}).

\subparagraph{Key Lemma.}
We show that we can bound the number of state switches
of the vertices of $G$
between two configurations separated by two
time steps.

\begin{restatable}{lemma}{keyLemma}
\label{lemma:bound_change}
For every $m \geq 3$,
the equation
$
\numbu{1 \wildcard 0}_m + \numbu{0 \wildcard 1}_m \leq 2|E| + 2(i_0+i_1+1)|V|
$
holds.
\end{restatable}

\subparagraph{Proof sketch.}
The basic idea
is that the current state of a vertex $s \in V$
indirectly contributes to $s$ having the same state
two steps in the future:
let us suppose that $s$
is alive at time $i \in \mathbb{N}$.
Then $s$ contributes towards
making its neighbours alive
at time $i+1$,
which in turn contribute towards making
their own neighbours, including $s$,
alive at time $i+2$.
We formalise this idea
by studying in details the number
of occurrences of diverse factors 
in the histories of $G$.
The full proof can be found in the appendix (Part \ref{appendix:lemma}).

\subsection{Proof of 
Proposition \ref{theorem:underpopulation_cycle}
and Proposition \ref{proposition:underpopulation}}\label{subsec_under_proof}

Using Lemma \ref{lemma:bound_change},
we are finally able to demonstrate the two results
left unproven at the beginning of this section.
First, we prove Proposition \ref{theorem:underpopulation_cycle}.
Note that the proof only uses the fact that
$\numbu{1\wildcard 0}_m + \numbu{0\wildcard 1}_m$
is bounded, and not the precise bound.

\theoremCycle*

\begin{proof}
Since Lemma \ref{lemma:bound_change}
bounds the value of
$\numbu{1\wildcard 0}_m + \numbu{0\wildcard 1}_m$
for every $m \in \mathbb{N}$,
then for every vertex $s \in V$,
factors of the form $1 \wildcard 0$ or $0 \wildcard 1$ 
only appear in a bounded prefix of $\word{s}$.
Therefore, the corresponding infinite suffix
only contains factors of the form
$1 \wildcard 1$ or $0 \wildcard 0$,
which immediately yields that the periodic part of $\word{s}$
is of size either $1$ or $2$.
Since this is verified by every vertex,
this shows that
under the underpopulation rule 
$\underpopulation{i_0}{i_1}$,
the graph $G$ either reaches a stable
configuration, or ends up alternating
between two configurations.
\end{proof}

Finally,
we prove Proposition \ref{proposition:underpopulation}.
This time, we actually need the precise bound
exposed by Lemma \ref{lemma:bound_change}.

\propositionReach*

\begin{proof}
We prove that if
no cycle is completed in
the first $2i$ steps of the process
for some $0 < i \in \mathbb{N}$,
then the histories of $G$ admit
at least $2i-2$
factors of the form $1\wildcard 0$ or $0 \wildcard 1$.
Since 
$\numbu{1\wildcard 0}_{2i+2} +
\numbu{0\wildcard 1}_{2i+2}$
is smaller than $2|E| + 2(i_0+i_1+1)|V|$
by Lemma \ref{lemma:bound_change},
this implies that
$2i \leq 2|E| + 2(i_0+i_1+1)|V|+2$,
which proves the lemma.

Let $i \in \mathbb{N}$ be a strictly positive integer,
and let us suppose that no configuration is repeated amongst
the first $2i$ steps of the process.
Let us first focus
on the sequence of $i$ odd configurations
$c_1,c_3,c_5,\ldots,c_{2i-1}$ of $G$.
By supposition, no configuration is repeated,
hence for every $1 \leq j \leq i-1$,
at least one vertex has distinct states
in the configurations $c_{2j-1}$
and $c_{2j+1}$.
These $i-1$ changes either consist 
in the death of a live vertex,
counting towards the value
$\numbu{1\wildcard 0}_{2i+2}$,
or in the birth of a dead vertex,
counting towards the value
$\numbu{0 \wildcard 1}_{2i+2}$.
Similarly, 
focusing on the sequence of $i$ even configurations
$c_2,c_4,c_6,\ldots,c_{2i}$
yields $i-1$ distinct occurrences
of vertices changing state between
two successive positions of even parity,
counting towards the value
$\numbu{1 \wildcard 0}_{2i+2} + \numbu{0 \wildcard 1}_{2i+2}$.
As a consequence,
$
2i-2 \leq \numbu{1 \wildcard 0}_{2i+2} + \numbu{0 \wildcard 1}_{2i+2}
$,
hence, by Lemma \ref{lemma:bound_change},
$
2i \leq 2|E| + 2(i_0+i_1+1)|V| + 2
$,
which concludes the proof.
\end{proof}

\section{Overpopulation: PSPACE completeness}
In this section, we study overpopulation rules:
a vertex comes to life if it has
sufficiently many living neighbors,
and dies if it has too many living neighbors.
Our result is in opposition to the result of the previous section:

\begin{restatable}{theorem}{theoremPSpaceHardness}\label{theorem:overpopulation}
  The following assertions hold:
  \begin{itemize}
\setlength{\itemsep}{0pt}
\setlength{\parskip}{0pt}
\setlength{\parsep}{0pt}
    \item For every overpopulation rule, the reachability and long-run average problems
      are in PSPACE.
    \item For the specific case $\overpopulation{2}{1}$, the reachability and
      long-run average are PSPACE-hard.
  \end{itemize}
\end{restatable}

\begin{remark}
  The first item of Theorem~\ref{theorem:overpopulation} (the PSPACE upper bound) is straightforward and presented
  in Appendix~\ref{appendix:overpopulationAlgorithm}.
  Our main contribution is item~$2$ (PSPACE hardness).
  We present a graph construction simulating a Turing machine.
  In addition, we show that our basic construction can be
  modified to ensure that we obtain a regular graph of degree $10$
  (details in Appendix~\ref{appendix:regularization}).
\end{remark}

\subsection{General idea for hardness}
We create a graph and an initial position simulating a polynomial-space Turing machine.
The graph is mostly composed of dead vertices,
with some live vertices that carry signals and store data.
The graph is regular and consists of \textbf{blobs} of vertices.
One blob corresponds to one cell of the tape and stores a tape alphabet symbol.
Blobs are connected in a row, and at most a single blob is active at every point in time.
The active blob receives a signal that corresponds to the state of the Turing machine.
It computes the transition function using the received signal and its stored value.
The result of the transition function is then used to
(1) modify the content of the blob
and (2) send the resulting state to the neighboring blob, activating it.

\begin{figure}
  \begin{minipage}{0.45\textwidth}
            \begin{center}
            \scalebox{0.5}{        
\begin{minipage}{.5\textwidth}
\begin{center}
\begin{tikzpicture}
\node[draw,circle,fill=black] (0Wire0) at (0, 0.0) {};
\node[draw,circle,fill=black] (0Wire1) at (0, 0.8) {};
\node[draw,circle,fill=black] (0Wire2) at (0, 1.6) {};
\node[draw,circle,fill=black] (0Wire3) at (0, 2.4000000000000004) {};
\node[draw,circle,fill=white] (1Wire0) at (1, 0.8) {};
\node[draw,circle,fill=white] (1Wire1) at (1, 1.6) {};
\node[draw,circle,fill=white] (2Wire0) at (2, 0.0) {};
\node[draw,circle,fill=white] (2Wire1) at (2, 0.8) {};
\node[draw,circle,fill=white] (2Wire2) at (2, 1.6) {};
\node[draw,circle,fill=white] (2Wire3) at (2, 2.4000000000000004) {};
\node[draw,circle,fill=white] (3Wire0) at (3, 0.8) {};
\node[draw,circle,fill=white] (3Wire1) at (3, 1.6) {};
\node[draw,circle,fill=white] (4Wire0) at (4, 0.0) {};
\node[draw,circle,fill=white] (4Wire1) at (4, 0.8) {};
\node[draw,circle,fill=white] (4Wire2) at (4, 1.6) {};
\node[draw,circle,fill=white] (4Wire3) at (4, 2.4000000000000004) {};
\draw (0Wire0) -- (0Wire0);
\draw (0Wire0) -- (0Wire0);
\draw (0Wire0) -- (1Wire0);
\draw (0Wire1) -- (0Wire1);
\draw (0Wire1) -- (0Wire1);
\draw (0Wire1) -- (1Wire0);
\draw (0Wire2) -- (0Wire2);
\draw (0Wire2) -- (0Wire2);
\draw (0Wire2) -- (1Wire1);
\draw (0Wire3) -- (0Wire3);
\draw (0Wire3) -- (0Wire3);
\draw (0Wire3) -- (1Wire1);
\draw (1Wire0) -- (0Wire0);
\draw (1Wire0) -- (0Wire1);
\draw (1Wire0) -- (2Wire0);
\draw (1Wire0) -- (2Wire1);
\draw (1Wire0) -- (2Wire2);
\draw (1Wire0) -- (2Wire3);
\draw (1Wire1) -- (0Wire2);
\draw (1Wire1) -- (0Wire3);
\draw (1Wire1) -- (2Wire0);
\draw (1Wire1) -- (2Wire1);
\draw (1Wire1) -- (2Wire2);
\draw (1Wire1) -- (2Wire3);
\draw (2Wire0) -- (1Wire0);
\draw (2Wire0) -- (1Wire1);
\draw (2Wire0) -- (3Wire0);
\draw (2Wire1) -- (1Wire0);
\draw (2Wire1) -- (1Wire1);
\draw (2Wire1) -- (3Wire0);
\draw (2Wire2) -- (1Wire0);
\draw (2Wire2) -- (1Wire1);
\draw (2Wire2) -- (3Wire1);
\draw (2Wire3) -- (1Wire0);
\draw (2Wire3) -- (1Wire1);
\draw (2Wire3) -- (3Wire1);
\draw (3Wire0) -- (2Wire0);
\draw (3Wire0) -- (2Wire1);
\draw (3Wire0) -- (4Wire0);
\draw (3Wire0) -- (4Wire1);
\draw (3Wire0) -- (4Wire2);
\draw (3Wire0) -- (4Wire3);
\draw (3Wire1) -- (2Wire2);
\draw (3Wire1) -- (2Wire3);
\draw (3Wire1) -- (4Wire0);
\draw (3Wire1) -- (4Wire1);
\draw (3Wire1) -- (4Wire2);
\draw (3Wire1) -- (4Wire3);
\draw (4Wire0) -- (3Wire0);
\draw (4Wire0) -- (3Wire1);
\draw (4Wire1) -- (3Wire0);
\draw (4Wire1) -- (3Wire1);
\draw (4Wire2) -- (3Wire0);
\draw (4Wire2) -- (3Wire1);
\draw (4Wire3) -- (3Wire0);
\draw (4Wire3) -- (3Wire1);
\end{tikzpicture}
\newline

\begin{tikzpicture}
\node[draw,circle,fill=white] (0Wire0) at (0, 0.0) {};
\node[draw,circle,fill=white] (0Wire1) at (0, 0.8) {};
\node[draw,circle,fill=white] (0Wire2) at (0, 1.6) {};
\node[draw,circle,fill=white] (0Wire3) at (0, 2.4000000000000004) {};
\node[draw,circle,fill=black] (1Wire0) at (1, 0.8) {};
\node[draw,circle,fill=black] (1Wire1) at (1, 1.6) {};
\node[draw,circle,fill=white] (2Wire0) at (2, 0.0) {};
\node[draw,circle,fill=white] (2Wire1) at (2, 0.8) {};
\node[draw,circle,fill=white] (2Wire2) at (2, 1.6) {};
\node[draw,circle,fill=white] (2Wire3) at (2, 2.4000000000000004) {};
\node[draw,circle,fill=white] (3Wire0) at (3, 0.8) {};
\node[draw,circle,fill=white] (3Wire1) at (3, 1.6) {};
\node[draw,circle,fill=white] (4Wire0) at (4, 0.0) {};
\node[draw,circle,fill=white] (4Wire1) at (4, 0.8) {};
\node[draw,circle,fill=white] (4Wire2) at (4, 1.6) {};
\node[draw,circle,fill=white] (4Wire3) at (4, 2.4000000000000004) {};
\draw (0Wire0) -- (0Wire0);
\draw (0Wire0) -- (0Wire0);
\draw (0Wire0) -- (1Wire0);
\draw (0Wire1) -- (0Wire1);
\draw (0Wire1) -- (0Wire1);
\draw (0Wire1) -- (1Wire0);
\draw (0Wire2) -- (0Wire2);
\draw (0Wire2) -- (0Wire2);
\draw (0Wire2) -- (1Wire1);
\draw (0Wire3) -- (0Wire3);
\draw (0Wire3) -- (0Wire3);
\draw (0Wire3) -- (1Wire1);
\draw (1Wire0) -- (0Wire0);
\draw (1Wire0) -- (0Wire1);
\draw (1Wire0) -- (2Wire0);
\draw (1Wire0) -- (2Wire1);
\draw (1Wire0) -- (2Wire2);
\draw (1Wire0) -- (2Wire3);
\draw (1Wire1) -- (0Wire2);
\draw (1Wire1) -- (0Wire3);
\draw (1Wire1) -- (2Wire0);
\draw (1Wire1) -- (2Wire1);
\draw (1Wire1) -- (2Wire2);
\draw (1Wire1) -- (2Wire3);
\draw (2Wire0) -- (1Wire0);
\draw (2Wire0) -- (1Wire1);
\draw (2Wire0) -- (3Wire0);
\draw (2Wire1) -- (1Wire0);
\draw (2Wire1) -- (1Wire1);
\draw (2Wire1) -- (3Wire0);
\draw (2Wire2) -- (1Wire0);
\draw (2Wire2) -- (1Wire1);
\draw (2Wire2) -- (3Wire1);
\draw (2Wire3) -- (1Wire0);
\draw (2Wire3) -- (1Wire1);
\draw (2Wire3) -- (3Wire1);
\draw (3Wire0) -- (2Wire0);
\draw (3Wire0) -- (2Wire1);
\draw (3Wire0) -- (4Wire0);
\draw (3Wire0) -- (4Wire1);
\draw (3Wire0) -- (4Wire2);
\draw (3Wire0) -- (4Wire3);
\draw (3Wire1) -- (2Wire2);
\draw (3Wire1) -- (2Wire3);
\draw (3Wire1) -- (4Wire0);
\draw (3Wire1) -- (4Wire1);
\draw (3Wire1) -- (4Wire2);
\draw (3Wire1) -- (4Wire3);
\draw (4Wire0) -- (3Wire0);
\draw (4Wire0) -- (3Wire1);
\draw (4Wire1) -- (3Wire0);
\draw (4Wire1) -- (3Wire1);
\draw (4Wire2) -- (3Wire0);
\draw (4Wire2) -- (3Wire1);
\draw (4Wire3) -- (3Wire0);
\draw (4Wire3) -- (3Wire1);
\end{tikzpicture}
\newline

\begin{tikzpicture}
\node[draw,circle,fill=white] (0Wire0) at (0, 0.0) {};
\node[draw,circle,fill=white] (0Wire1) at (0, 0.8) {};
\node[draw,circle,fill=white] (0Wire2) at (0, 1.6) {};
\node[draw,circle,fill=white] (0Wire3) at (0, 2.4000000000000004) {};
\node[draw,circle,fill=black] (1Wire0) at (1, 0.8) {};
\node[draw,circle,fill=black] (1Wire1) at (1, 1.6) {};
\node[draw,circle,fill=black] (2Wire0) at (2, 0.0) {};
\node[draw,circle,fill=black] (2Wire1) at (2, 0.8) {};
\node[draw,circle,fill=black] (2Wire2) at (2, 1.6) {};
\node[draw,circle,fill=black] (2Wire3) at (2, 2.4000000000000004) {};
\node[draw,circle,fill=white] (3Wire0) at (3, 0.8) {};
\node[draw,circle,fill=white] (3Wire1) at (3, 1.6) {};
\node[draw,circle,fill=white] (4Wire0) at (4, 0.0) {};
\node[draw,circle,fill=white] (4Wire1) at (4, 0.8) {};
\node[draw,circle,fill=white] (4Wire2) at (4, 1.6) {};
\node[draw,circle,fill=white] (4Wire3) at (4, 2.4000000000000004) {};
\draw (0Wire0) -- (0Wire0);
\draw (0Wire0) -- (0Wire0);
\draw (0Wire0) -- (1Wire0);
\draw (0Wire1) -- (0Wire1);
\draw (0Wire1) -- (0Wire1);
\draw (0Wire1) -- (1Wire0);
\draw (0Wire2) -- (0Wire2);
\draw (0Wire2) -- (0Wire2);
\draw (0Wire2) -- (1Wire1);
\draw (0Wire3) -- (0Wire3);
\draw (0Wire3) -- (0Wire3);
\draw (0Wire3) -- (1Wire1);
\draw (1Wire0) -- (0Wire0);
\draw (1Wire0) -- (0Wire1);
\draw (1Wire0) -- (2Wire0);
\draw (1Wire0) -- (2Wire1);
\draw (1Wire0) -- (2Wire2);
\draw (1Wire0) -- (2Wire3);
\draw (1Wire1) -- (0Wire2);
\draw (1Wire1) -- (0Wire3);
\draw (1Wire1) -- (2Wire0);
\draw (1Wire1) -- (2Wire1);
\draw (1Wire1) -- (2Wire2);
\draw (1Wire1) -- (2Wire3);
\draw (2Wire0) -- (1Wire0);
\draw (2Wire0) -- (1Wire1);
\draw (2Wire0) -- (3Wire0);
\draw (2Wire1) -- (1Wire0);
\draw (2Wire1) -- (1Wire1);
\draw (2Wire1) -- (3Wire0);
\draw (2Wire2) -- (1Wire0);
\draw (2Wire2) -- (1Wire1);
\draw (2Wire2) -- (3Wire1);
\draw (2Wire3) -- (1Wire0);
\draw (2Wire3) -- (1Wire1);
\draw (2Wire3) -- (3Wire1);
\draw (3Wire0) -- (2Wire0);
\draw (3Wire0) -- (2Wire1);
\draw (3Wire0) -- (4Wire0);
\draw (3Wire0) -- (4Wire1);
\draw (3Wire0) -- (4Wire2);
\draw (3Wire0) -- (4Wire3);
\draw (3Wire1) -- (2Wire2);
\draw (3Wire1) -- (2Wire3);
\draw (3Wire1) -- (4Wire0);
\draw (3Wire1) -- (4Wire1);
\draw (3Wire1) -- (4Wire2);
\draw (3Wire1) -- (4Wire3);
\draw (4Wire0) -- (3Wire0);
\draw (4Wire0) -- (3Wire1);
\draw (4Wire1) -- (3Wire0);
\draw (4Wire1) -- (3Wire1);
\draw (4Wire2) -- (3Wire0);
\draw (4Wire2) -- (3Wire1);
\draw (4Wire3) -- (3Wire0);
\draw (4Wire3) -- (3Wire1);
\end{tikzpicture}
\newline

\begin{tikzpicture}
\node[draw,circle,fill=white] (0Wire0) at (0, 0.0) {};
\node[draw,circle,fill=white] (0Wire1) at (0, 0.8) {};
\node[draw,circle,fill=white] (0Wire2) at (0, 1.6) {};
\node[draw,circle,fill=white] (0Wire3) at (0, 2.4000000000000004) {};
\node[draw,circle,fill=white] (1Wire0) at (1, 0.8) {};
\node[draw,circle,fill=white] (1Wire1) at (1, 1.6) {};
\node[draw,circle,fill=white] (2Wire0) at (2, 0.0) {};
\node[draw,circle,fill=white] (2Wire1) at (2, 0.8) {};
\node[draw,circle,fill=white] (2Wire2) at (2, 1.6) {};
\node[draw,circle,fill=white] (2Wire3) at (2, 2.4000000000000004) {};
\node[draw,circle,fill=black] (3Wire0) at (3, 0.8) {};
\node[draw,circle,fill=black] (3Wire1) at (3, 1.6) {};
\node[draw,circle,fill=white] (4Wire0) at (4, 0.0) {};
\node[draw,circle,fill=white] (4Wire1) at (4, 0.8) {};
\node[draw,circle,fill=white] (4Wire2) at (4, 1.6) {};
\node[draw,circle,fill=white] (4Wire3) at (4, 2.4000000000000004) {};
\draw (0Wire0) -- (0Wire0);
\draw (0Wire0) -- (0Wire0);
\draw (0Wire0) -- (1Wire0);
\draw (0Wire1) -- (0Wire1);
\draw (0Wire1) -- (0Wire1);
\draw (0Wire1) -- (1Wire0);
\draw (0Wire2) -- (0Wire2);
\draw (0Wire2) -- (0Wire2);
\draw (0Wire2) -- (1Wire1);
\draw (0Wire3) -- (0Wire3);
\draw (0Wire3) -- (0Wire3);
\draw (0Wire3) -- (1Wire1);
\draw (1Wire0) -- (0Wire0);
\draw (1Wire0) -- (0Wire1);
\draw (1Wire0) -- (2Wire0);
\draw (1Wire0) -- (2Wire1);
\draw (1Wire0) -- (2Wire2);
\draw (1Wire0) -- (2Wire3);
\draw (1Wire1) -- (0Wire2);
\draw (1Wire1) -- (0Wire3);
\draw (1Wire1) -- (2Wire0);
\draw (1Wire1) -- (2Wire1);
\draw (1Wire1) -- (2Wire2);
\draw (1Wire1) -- (2Wire3);
\draw (2Wire0) -- (1Wire0);
\draw (2Wire0) -- (1Wire1);
\draw (2Wire0) -- (3Wire0);
\draw (2Wire1) -- (1Wire0);
\draw (2Wire1) -- (1Wire1);
\draw (2Wire1) -- (3Wire0);
\draw (2Wire2) -- (1Wire0);
\draw (2Wire2) -- (1Wire1);
\draw (2Wire2) -- (3Wire1);
\draw (2Wire3) -- (1Wire0);
\draw (2Wire3) -- (1Wire1);
\draw (2Wire3) -- (3Wire1);
\draw (3Wire0) -- (2Wire0);
\draw (3Wire0) -- (2Wire1);
\draw (3Wire0) -- (4Wire0);
\draw (3Wire0) -- (4Wire1);
\draw (3Wire0) -- (4Wire2);
\draw (3Wire0) -- (4Wire3);
\draw (3Wire1) -- (2Wire2);
\draw (3Wire1) -- (2Wire3);
\draw (3Wire1) -- (4Wire0);
\draw (3Wire1) -- (4Wire1);
\draw (3Wire1) -- (4Wire2);
\draw (3Wire1) -- (4Wire3);
\draw (4Wire0) -- (3Wire0);
\draw (4Wire0) -- (3Wire1);
\draw (4Wire1) -- (3Wire0);
\draw (4Wire1) -- (3Wire1);
\draw (4Wire2) -- (3Wire0);
\draw (4Wire2) -- (3Wire1);
\draw (4Wire3) -- (3Wire0);
\draw (4Wire3) -- (3Wire1);
\end{tikzpicture}
\newline

\begin{tikzpicture}
\node[draw,circle,fill=white] (0Wire0) at (0, 0.0) {};
\node[draw,circle,fill=white] (0Wire1) at (0, 0.8) {};
\node[draw,circle,fill=white] (0Wire2) at (0, 1.6) {};
\node[draw,circle,fill=white] (0Wire3) at (0, 2.4000000000000004) {};
\node[draw,circle,fill=white] (1Wire0) at (1, 0.8) {};
\node[draw,circle,fill=white] (1Wire1) at (1, 1.6) {};
\node[draw,circle,fill=white] (2Wire0) at (2, 0.0) {};
\node[draw,circle,fill=white] (2Wire1) at (2, 0.8) {};
\node[draw,circle,fill=white] (2Wire2) at (2, 1.6) {};
\node[draw,circle,fill=white] (2Wire3) at (2, 2.4000000000000004) {};
\node[draw,circle,fill=black] (3Wire0) at (3, 0.8) {};
\node[draw,circle,fill=black] (3Wire1) at (3, 1.6) {};
\node[draw,circle,fill=black] (4Wire0) at (4, 0.0) {};
\node[draw,circle,fill=black] (4Wire1) at (4, 0.8) {};
\node[draw,circle,fill=black] (4Wire2) at (4, 1.6) {};
\node[draw,circle,fill=black] (4Wire3) at (4, 2.4000000000000004) {};
\draw (0Wire0) -- (0Wire0);
\draw (0Wire0) -- (0Wire0);
\draw (0Wire0) -- (1Wire0);
\draw (0Wire1) -- (0Wire1);
\draw (0Wire1) -- (0Wire1);
\draw (0Wire1) -- (1Wire0);
\draw (0Wire2) -- (0Wire2);
\draw (0Wire2) -- (0Wire2);
\draw (0Wire2) -- (1Wire1);
\draw (0Wire3) -- (0Wire3);
\draw (0Wire3) -- (0Wire3);
\draw (0Wire3) -- (1Wire1);
\draw (1Wire0) -- (0Wire0);
\draw (1Wire0) -- (0Wire1);
\draw (1Wire0) -- (2Wire0);
\draw (1Wire0) -- (2Wire1);
\draw (1Wire0) -- (2Wire2);
\draw (1Wire0) -- (2Wire3);
\draw (1Wire1) -- (0Wire2);
\draw (1Wire1) -- (0Wire3);
\draw (1Wire1) -- (2Wire0);
\draw (1Wire1) -- (2Wire1);
\draw (1Wire1) -- (2Wire2);
\draw (1Wire1) -- (2Wire3);
\draw (2Wire0) -- (1Wire0);
\draw (2Wire0) -- (1Wire1);
\draw (2Wire0) -- (3Wire0);
\draw (2Wire1) -- (1Wire0);
\draw (2Wire1) -- (1Wire1);
\draw (2Wire1) -- (3Wire0);
\draw (2Wire2) -- (1Wire0);
\draw (2Wire2) -- (1Wire1);
\draw (2Wire2) -- (3Wire1);
\draw (2Wire3) -- (1Wire0);
\draw (2Wire3) -- (1Wire1);
\draw (2Wire3) -- (3Wire1);
\draw (3Wire0) -- (2Wire0);
\draw (3Wire0) -- (2Wire1);
\draw (3Wire0) -- (4Wire0);
\draw (3Wire0) -- (4Wire1);
\draw (3Wire0) -- (4Wire2);
\draw (3Wire0) -- (4Wire3);
\draw (3Wire1) -- (2Wire2);
\draw (3Wire1) -- (2Wire3);
\draw (3Wire1) -- (4Wire0);
\draw (3Wire1) -- (4Wire1);
\draw (3Wire1) -- (4Wire2);
\draw (3Wire1) -- (4Wire3);
\draw (4Wire0) -- (3Wire0);
\draw (4Wire0) -- (3Wire1);
\draw (4Wire1) -- (3Wire0);
\draw (4Wire1) -- (3Wire1);
\draw (4Wire2) -- (3Wire0);
\draw (4Wire2) -- (3Wire1);
\draw (4Wire3) -- (3Wire0);
\draw (4Wire3) -- (3Wire1);
\end{tikzpicture}
\newline

\end{center}
\end{minipage}\hfill\break
            }
            \end{center}
            \caption{ Signal going throught two connected wires (We suppose its left end is output of some gadget).}\label{livewire}
  \end{minipage}
  \hspace{0.1\textwidth}
  \begin{minipage}{0.45\textwidth}
    
            %\begin{figure}
            \begin{center}
            \scalebox{0.5}{
            
\begin{tikzpicture}
\node (Andanno1) at (-1.3, 5.2) {  \LARGE $I_1$};
\node[scale=5,color=black!30!white] (Andanno1) at (-0.4, 5.2) {  \LARGE \textbraceleft };
\node (Andanno2) at (-1.3, 1.2) {  \LARGE $I_2$};
\node[scale=5,color=black!30!white] (Andanno2) at (-0.4, 1.2) {  \LARGE \textbraceleft };
\node (Andanno3) at (5.45, 3.2) {  \LARGE $O_1$};
\node[scale=5,color=black!30!white] (Andanno3) at (4.4, 3.2) {  \LARGE \textbraceright };
\node[draw,circle,fill=white] (0Andcror0) at (0, 0.0) {};
\node[draw,circle,fill=white] (0Andcror1) at (0, 0.8) {};
\node[draw,circle,fill=white] (0Andcror2) at (0, 1.6) {};
\node[draw,circle,fill=white] (0Andcror3) at (0, 2.4000000000000004) {};
\node[draw,circle,fill=white] (1Andcror0) at (1, 0.8) {};
\node[draw,circle,fill=white] (1Andcror1) at (1, 1.6) {};
\node[draw,circle,fill=white] (0Andcrora0) at (0, 4.0) {};
\node[draw,circle,fill=white] (0Andcrora1) at (0, 4.8) {};
\node[draw,circle,fill=white] (0Andcrora2) at (0, 5.6) {};
\node[draw,circle,fill=white] (0Andcrora3) at (0, 6.4) {};
\node[draw,circle,fill=white] (1Andcrora0) at (1, 4.8) {};
\node[draw,circle,fill=white] (1Andcrora1) at (1, 5.6) {};
\node[draw,circle,fill=white] (0ao0) at (2, 2.0) {};
\node[draw,circle,fill=white] (0ao1) at (2, 2.8) {};
\node[draw,circle,fill=white] (0ao2) at (2, 3.6) {};
\node[draw,circle,fill=white] (0ao3) at (2, 4.4) {};
\node[draw,circle,fill=white] (1ao0) at (3, 2.8) {};
\node[draw,circle,fill=white] (1ao1) at (3, 3.6) {};
\node[draw,circle,fill=white] (2ao0) at (4, 2.0) {};
\node[draw,circle,fill=white] (2ao1) at (4, 2.8) {};
\node[draw,circle,fill=white] (2ao2) at (4, 3.6) {};
\node[draw,circle,fill=white] (2ao3) at (4, 4.4) {};
\draw (0Andcror0) -- (0Andcror0);
\draw (0Andcror0) -- (0Andcror0);
\draw (0Andcror0) -- (1Andcror0);
\draw (0Andcror1) -- (0Andcror1);
\draw (0Andcror1) -- (0Andcror1);
\draw (0Andcror1) -- (1Andcror0);
\draw (0Andcror2) -- (0Andcror2);
\draw (0Andcror2) -- (0Andcror2);
\draw (0Andcror2) -- (1Andcror1);
\draw (0Andcror3) -- (0Andcror3);
\draw (0Andcror3) -- (0Andcror3);
\draw (0Andcror3) -- (1Andcror1);
\draw (1Andcror0) -- (0Andcror0);
\draw (1Andcror0) -- (0Andcror1);
\draw (1Andcror0) -- (0ao0);
\draw (1Andcror0) -- (0ao2);
\draw (1Andcror1) -- (0Andcror2);
\draw (1Andcror1) -- (0Andcror3);
\draw (1Andcror1) -- (0ao0);
\draw (1Andcror1) -- (0ao2);
\draw (0Andcrora0) -- (0Andcrora0);
\draw (0Andcrora0) -- (0Andcrora0);
\draw (0Andcrora0) -- (1Andcrora0);
\draw (0Andcrora1) -- (0Andcrora1);
\draw (0Andcrora1) -- (0Andcrora1);
\draw (0Andcrora1) -- (1Andcrora0);
\draw (0Andcrora2) -- (0Andcrora2);
\draw (0Andcrora2) -- (0Andcrora2);
\draw (0Andcrora2) -- (1Andcrora1);
\draw (0Andcrora3) -- (0Andcrora3);
\draw (0Andcrora3) -- (0Andcrora3);
\draw (0Andcrora3) -- (1Andcrora1);
\draw (1Andcrora0) -- (0Andcrora0);
\draw (1Andcrora0) -- (0Andcrora1);
\draw (1Andcrora0) -- (0ao1);
\draw (1Andcrora0) -- (0ao3);
\draw (1Andcrora1) -- (0Andcrora2);
\draw (1Andcrora1) -- (0Andcrora3);
\draw (1Andcrora1) -- (0ao1);
\draw (1Andcrora1) -- (0ao3);
\draw (0ao0) -- (0ao0);
\draw (0ao0) -- (0ao0);
\draw (0ao0) -- (1ao0);
\draw (0ao0) -- (1Andcror0);
\draw (0ao0) -- (1Andcror1);
\draw (0ao1) -- (0ao1);
\draw (0ao1) -- (0ao1);
\draw (0ao1) -- (1ao0);
\draw (0ao1) -- (1Andcrora0);
\draw (0ao1) -- (1Andcrora1);
\draw (0ao2) -- (0ao2);
\draw (0ao2) -- (0ao2);
\draw (0ao2) -- (1ao1);
\draw (0ao2) -- (1Andcror0);
\draw (0ao2) -- (1Andcror1);
\draw (0ao3) -- (0ao3);
\draw (0ao3) -- (0ao3);
\draw (0ao3) -- (1ao1);
\draw (0ao3) -- (1Andcrora0);
\draw (0ao3) -- (1Andcrora1);
\draw (1ao0) -- (0ao0);
\draw (1ao0) -- (0ao1);
\draw (1ao0) -- (2ao0);
\draw (1ao0) -- (2ao1);
\draw (1ao0) -- (2ao2);
\draw (1ao0) -- (2ao3);
\draw (1ao1) -- (0ao2);
\draw (1ao1) -- (0ao3);
\draw (1ao1) -- (2ao0);
\draw (1ao1) -- (2ao1);
\draw (1ao1) -- (2ao2);
\draw (1ao1) -- (2ao3);
\draw (2ao0) -- (1ao0);
\draw (2ao0) -- (1ao1);
\draw (2ao1) -- (1ao0);
\draw (2ao1) -- (1ao1);
\draw (2ao2) -- (1ao0);
\draw (2ao2) -- (1ao1);
\draw (2ao3) -- (1ao0);
\draw (2ao3) -- (1ao1);
\end{tikzpicture}

            }
            \end{center}
            \caption{ Gate computing AND.}\label{andpicture}
            %\end{figure}

            %\begin{figure}
            \begin{center}
            \scalebox{0.5}{
            
\begin{tikzpicture}
\node (Notanno1) at (-1.3, 4.2) {  \LARGE $I_1$};
\node[scale=5,color=black!30!white] (Notanno1) at (-0.4, 4.2) {  \LARGE \textbraceleft };
\node (Notanno2) at (-1.3, 0.2) {  \LARGE $I_2$};
\node[scale=5,color=black!30!white] (Notanno2) at (-0.4, 0.2) {  \LARGE \textbraceleft };
\node (Notanno3) at (5.8, 3.7) {  \LARGE $O_1$};
\node[scale=5,color=black!30!white] (Notanno3) at (4.75, 3.7) {  \LARGE \textbraceright };
\node[draw,circle,fill=white] (0NotsplitterSUlev1) at (2, 2.0) {};
\node[draw,circle,fill=white] (1NotsplitterSUlev1) at (2, 2.5) {};
\node[draw,circle,fill=white] (2NotsplitterSUlev1) at (2, 3.0) {};
\node[draw,circle,fill=white] (3NotsplitterSUlev1) at (2, 3.5) {};
\node[draw,circle,fill=white] (0NotsplitterSUlev2) at (2, 4.5) {};
\node[draw,circle,fill=white] (1NotsplitterSUlev2) at (2, 5.0) {};
\node[draw,circle,fill=white] (2NotsplitterSUlev2) at (2, 5.5) {};
\node[draw,circle,fill=white] (3NotsplitterSUlev2) at (2, 6.0) {};
\node[draw,circle,fill=white] (0NotsplitterSUw0) at (0, 3.0) {};
\node[draw,circle,fill=white] (0NotsplitterSUw1) at (0, 3.8) {};
\node[draw,circle,fill=white] (0NotsplitterSUw2) at (0, 4.6) {};
\node[draw,circle,fill=white] (0NotsplitterSUw3) at (0, 5.4) {};
\node[draw,circle,fill=white] (1NotsplitterSUw0) at (1, 3.8) {};
\node[draw,circle,fill=white] (1NotsplitterSUw1) at (1, 4.6) {};
\node[draw,circle,fill=white] (0NotSUfw0) at (3, 1.8) {};
\node[draw,circle,fill=white] (0NotSUfw1) at (3, 2.6) {};
\node[draw,circle,fill=white] (0NotSUsw0) at (3, 3.8) {};
\node[draw,circle,fill=white] (0NotSUsw1) at (3, 4.6) {};
\node[draw,circle,fill=white] (0Notsulow0) at (0, -1.0) {};
\node[draw,circle,fill=white] (0Notsulow1) at (0, -0.19999999999999996) {};
\node[draw,circle,fill=white] (0Notsulow2) at (0, 0.6000000000000001) {};
\node[draw,circle,fill=white] (0Notsulow3) at (0, 1.4000000000000004) {};
\node[draw,circle,fill=white] (1Notsulow0) at (1, -0.19999999999999996) {};
\node[draw,circle,fill=white] (1Notsulow1) at (1, 0.6000000000000001) {};
\node[draw,circle,fill=white] (2Notsulow0) at (2, -1.0) {};
\node[draw,circle,fill=white] (2Notsulow1) at (2, -0.19999999999999996) {};
\node[draw,circle,fill=white] (2Notsulow2) at (2, 0.6000000000000001) {};
\node[draw,circle,fill=white] (2Notsulow3) at (2, 1.4000000000000004) {};
\node[draw,circle,fill=white] (3Notsulow0) at (3, -0.19999999999999996) {};
\node[draw,circle,fill=white] (3Notsulow1) at (3, 0.6000000000000001) {};
\node[draw,circle,fill=white] (0Notsulast0) at (4.3, 2.5) {};
\node[draw,circle,fill=white] (0Notsulast1) at (4.3, 3.3) {};
\node[draw,circle,fill=white] (0Notsulast2) at (4.3, 4.1) {};
\node[draw,circle,fill=white] (0Notsulast3) at (4.3, 4.9) {};
\node[draw,circle,fill=white] (0Notsubalance0) at (4, -1.0) {};
\node[draw,circle,fill=white] (0Notsubalance1) at (4, -0.19999999999999996) {};
\node[draw,circle,fill=white] (0Notsubalance2) at (4, 0.6000000000000001) {};
\node[draw,circle,fill=white] (0Notsubalance3) at (4, 1.4000000000000004) {};
\draw (0NotsplitterSUlev1) -- (1NotsplitterSUw0);
\draw (0NotsplitterSUlev1) -- (1NotsplitterSUw1);
\draw (0NotsplitterSUlev1) -- (0NotSUfw0);
\draw (1NotsplitterSUlev1) -- (1NotsplitterSUw0);
\draw (1NotsplitterSUlev1) -- (1NotsplitterSUw1);
\draw (1NotsplitterSUlev1) -- (0NotSUfw0);
\draw (2NotsplitterSUlev1) -- (1NotsplitterSUw0);
\draw (2NotsplitterSUlev1) -- (1NotsplitterSUw1);
\draw (2NotsplitterSUlev1) -- (0NotSUfw1);
\draw (3NotsplitterSUlev1) -- (1NotsplitterSUw0);
\draw (3NotsplitterSUlev1) -- (1NotsplitterSUw1);
\draw (3NotsplitterSUlev1) -- (0NotSUfw1);
\draw (0NotsplitterSUlev2) -- (1NotsplitterSUw0);
\draw (0NotsplitterSUlev2) -- (1NotsplitterSUw1);
\draw (0NotsplitterSUlev2) -- (0NotSUsw0);
\draw (1NotsplitterSUlev2) -- (1NotsplitterSUw0);
\draw (1NotsplitterSUlev2) -- (1NotsplitterSUw1);
\draw (1NotsplitterSUlev2) -- (0NotSUsw0);
\draw (2NotsplitterSUlev2) -- (1NotsplitterSUw0);
\draw (2NotsplitterSUlev2) -- (1NotsplitterSUw1);
\draw (2NotsplitterSUlev2) -- (0NotSUsw1);
\draw (3NotsplitterSUlev2) -- (1NotsplitterSUw0);
\draw (3NotsplitterSUlev2) -- (1NotsplitterSUw1);
\draw (3NotsplitterSUlev2) -- (0NotSUsw1);
\draw (0NotsplitterSUw0) -- (0NotsplitterSUw0);
\draw (0NotsplitterSUw0) -- (0NotsplitterSUw0);
\draw (0NotsplitterSUw0) -- (1NotsplitterSUw0);
\draw (0NotsplitterSUw1) -- (0NotsplitterSUw1);
\draw (0NotsplitterSUw1) -- (0NotsplitterSUw1);
\draw (0NotsplitterSUw1) -- (1NotsplitterSUw0);
\draw (0NotsplitterSUw2) -- (0NotsplitterSUw2);
\draw (0NotsplitterSUw2) -- (0NotsplitterSUw2);
\draw (0NotsplitterSUw2) -- (1NotsplitterSUw1);
\draw (0NotsplitterSUw3) -- (0NotsplitterSUw3);
\draw (0NotsplitterSUw3) -- (0NotsplitterSUw3);
\draw (0NotsplitterSUw3) -- (1NotsplitterSUw1);
\draw (1NotsplitterSUw0) -- (0NotsplitterSUw0);
\draw (1NotsplitterSUw0) -- (0NotsplitterSUw1);
\draw (1NotsplitterSUw0) -- (0NotsplitterSUlev1);
\draw (1NotsplitterSUw0) -- (1NotsplitterSUlev1);
\draw (1NotsplitterSUw0) -- (2NotsplitterSUlev1);
\draw (1NotsplitterSUw0) -- (3NotsplitterSUlev1);
\draw (1NotsplitterSUw0) -- (0NotsplitterSUlev2);
\draw (1NotsplitterSUw0) -- (1NotsplitterSUlev2);
\draw (1NotsplitterSUw0) -- (2NotsplitterSUlev2);
\draw (1NotsplitterSUw0) -- (3NotsplitterSUlev2);
\draw (1NotsplitterSUw1) -- (0NotsplitterSUw2);
\draw (1NotsplitterSUw1) -- (0NotsplitterSUw3);
\draw (1NotsplitterSUw1) -- (0NotsplitterSUlev1);
\draw (1NotsplitterSUw1) -- (1NotsplitterSUlev1);
\draw (1NotsplitterSUw1) -- (2NotsplitterSUlev1);
\draw (1NotsplitterSUw1) -- (3NotsplitterSUlev1);
\draw (1NotsplitterSUw1) -- (0NotsplitterSUlev2);
\draw (1NotsplitterSUw1) -- (1NotsplitterSUlev2);
\draw (1NotsplitterSUw1) -- (2NotsplitterSUlev2);
\draw (1NotsplitterSUw1) -- (3NotsplitterSUlev2);
\draw (0NotSUfw0) -- (0NotSUfw0);
\draw (0NotSUfw0) -- (0NotSUfw0);
\draw (0NotSUfw0) -- (0NotsplitterSUlev1);
\draw (0NotSUfw0) -- (1NotsplitterSUlev1);
\draw (0NotSUfw0) -- (0Notsulast0);
\draw (0NotSUfw1) -- (0NotSUfw1);
\draw (0NotSUfw1) -- (0NotSUfw1);
\draw (0NotSUfw1) -- (2NotsplitterSUlev1);
\draw (0NotSUfw1) -- (3NotsplitterSUlev1);
\draw (0NotSUfw1) -- (0Notsulast1);
\draw (0NotSUsw0) -- (0NotSUsw0);
\draw (0NotSUsw0) -- (0NotSUsw0);
\draw (0NotSUsw0) -- (0NotsplitterSUlev2);
\draw (0NotSUsw0) -- (1NotsplitterSUlev2);
\draw (0NotSUsw0) -- (0Notsulast2);
\draw (0NotSUsw1) -- (0NotSUsw1);
\draw (0NotSUsw1) -- (0NotSUsw1);
\draw (0NotSUsw1) -- (2NotsplitterSUlev2);
\draw (0NotSUsw1) -- (3NotsplitterSUlev2);
\draw (0NotSUsw1) -- (0Notsulast3);
\draw (0Notsulow0) -- (0Notsulow0);
\draw (0Notsulow0) -- (0Notsulow0);
\draw (0Notsulow0) -- (1Notsulow0);
\draw (0Notsulow1) -- (0Notsulow1);
\draw (0Notsulow1) -- (0Notsulow1);
\draw (0Notsulow1) -- (1Notsulow0);
\draw (0Notsulow2) -- (0Notsulow2);
\draw (0Notsulow2) -- (0Notsulow2);
\draw (0Notsulow2) -- (1Notsulow1);
\draw (0Notsulow3) -- (0Notsulow3);
\draw (0Notsulow3) -- (0Notsulow3);
\draw (0Notsulow3) -- (1Notsulow1);
\draw (1Notsulow0) -- (0Notsulow0);
\draw (1Notsulow0) -- (0Notsulow1);
\draw (1Notsulow0) -- (2Notsulow0);
\draw (1Notsulow0) -- (2Notsulow1);
\draw (1Notsulow0) -- (2Notsulow2);
\draw (1Notsulow0) -- (2Notsulow3);
\draw (1Notsulow1) -- (0Notsulow2);
\draw (1Notsulow1) -- (0Notsulow3);
\draw (1Notsulow1) -- (2Notsulow0);
\draw (1Notsulow1) -- (2Notsulow1);
\draw (1Notsulow1) -- (2Notsulow2);
\draw (1Notsulow1) -- (2Notsulow3);
\draw (2Notsulow0) -- (1Notsulow0);
\draw (2Notsulow0) -- (1Notsulow1);
\draw (2Notsulow0) -- (3Notsulow0);
\draw (2Notsulow1) -- (1Notsulow0);
\draw (2Notsulow1) -- (1Notsulow1);
\draw (2Notsulow1) -- (3Notsulow0);
\draw (2Notsulow2) -- (1Notsulow0);
\draw (2Notsulow2) -- (1Notsulow1);
\draw (2Notsulow2) -- (3Notsulow1);
\draw (2Notsulow3) -- (1Notsulow0);
\draw (2Notsulow3) -- (1Notsulow1);
\draw (2Notsulow3) -- (3Notsulow1);
\draw (3Notsulow0) -- (2Notsulow0);
\draw (3Notsulow0) -- (2Notsulow1);
\draw (3Notsulow0) -- (0Notsulast0);
\draw (3Notsulow0) -- (0Notsulast1);
\draw (3Notsulow0) -- (0Notsubalance0);
\draw (3Notsulow0) -- (0Notsubalance1);
\draw (3Notsulow0) -- (0Notsubalance2);
\draw (3Notsulow0) -- (0Notsubalance3);
\draw (3Notsulow1) -- (2Notsulow2);
\draw (3Notsulow1) -- (2Notsulow3);
\draw (3Notsulow1) -- (0Notsulast2);
\draw (3Notsulow1) -- (0Notsulast3);
\draw (3Notsulow1) -- (0Notsubalance0);
\draw (3Notsulow1) -- (0Notsubalance1);
\draw (3Notsulow1) -- (0Notsubalance2);
\draw (3Notsulow1) -- (0Notsubalance3);
\draw (0Notsulast0) -- (0Notsulast0);
\draw (0Notsulast0) -- (0Notsulast0);
\draw (0Notsulast0) -- (3Notsulow0);
\draw (0Notsulast0) -- (0NotSUfw0);
\draw (0Notsulast1) -- (0Notsulast1);
\draw (0Notsulast1) -- (0Notsulast1);
\draw (0Notsulast1) -- (3Notsulow0);
\draw (0Notsulast1) -- (0NotSUfw1);
\draw (0Notsulast2) -- (0Notsulast2);
\draw (0Notsulast2) -- (0Notsulast2);
\draw (0Notsulast2) -- (3Notsulow1);
\draw (0Notsulast2) -- (0NotSUsw0);
\draw (0Notsulast3) -- (0Notsulast3);
\draw (0Notsulast3) -- (0Notsulast3);
\draw (0Notsulast3) -- (3Notsulow1);
\draw (0Notsulast3) -- (0NotSUsw1);
\draw (0Notsubalance0) -- (0Notsubalance0);
\draw (0Notsubalance0) -- (0Notsubalance0);
\draw (0Notsubalance0) -- (3Notsulow0);
\draw (0Notsubalance0) -- (3Notsulow1);
\draw (0Notsubalance1) -- (0Notsubalance1);
\draw (0Notsubalance1) -- (0Notsubalance1);
\draw (0Notsubalance1) -- (3Notsulow0);
\draw (0Notsubalance1) -- (3Notsulow1);
\draw (0Notsubalance2) -- (0Notsubalance2);
\draw (0Notsubalance2) -- (0Notsubalance2);
\draw (0Notsubalance2) -- (3Notsulow0);
\draw (0Notsubalance2) -- (3Notsulow1);
\draw (0Notsubalance3) -- (0Notsubalance3);
\draw (0Notsubalance3) -- (0Notsubalance3);
\draw (0Notsubalance3) -- (3Notsulow0);
\draw (0Notsubalance3) -- (3Notsulow1);
\end{tikzpicture}

            }
            \end{center}
            \caption{ Storage Unit, signal at $I_1$ toggles state of four vertices to the left of $O_1$. Signal at $I_2$ gets to $O_1$ only if these four vertices are alive.}\label{supicture}
            %\end{figure}
  \end{minipage}
\end{figure}

\subsection{Basic gadgets}
We describe the gadgets used in the construction.
Each gadget $g$ has a constant number of inputs $I_1,I_2,\dots,I_c$
and outputs $O_1,O_2,\dots,O_d$ that receive, respectively send, signals.
Each input $I_i$ (output $O_i$)
is composed of four vertices that always share the same state.
We view live cells as true and dead cells as false,
and denote by $I_{i,t} \in \{0,1\}$ ($O_{i,t} \in \{0,1\}$)
the value of the input $I_i$ (output $O_i$) at time $t$.
Each one of our basic gadgets $g$ has an evaluation time $t_g$,
and is determined to realize a function
$f_g: \{0,1\}^c \rightarrow \{0,1\}^d$.
Starting from the \emph{inert} state (i.e, all the vertices are dead),
if the $c$ inputs receive some signal
(and no new parasite signal is received during the next $t_g$ steps of the process),
it computes $f_g$ in $t_g$ steps,
broadcast the result through the $d$ outputs,
and then goes back to the inert state.
We say that $g$ \emph{computes} the function
\[
\begin{array}{lclc}
f_g: & \{0,1\}^c & \rightarrow & \{0,1\}^d,\\
& (I_{1,t},I_{2,t},\dots,I_{c,t})
& \mapsto & (O_{1,t+t_g},O_{2,t+t_g},\dots,O_{d,t+t_g}).
\end{array}
\]
Moreover, for each gadget we suppose that the input is erased
after one step, and in turn the gadget is responsible for erasing its output after one step.
Here are the \textbf{basic gadgets}:
\begin{itemize}
  \item The wire
  transmits a signal.
  It is evaluated in $2$ time steps,
  has one input $I_1$, and one output $O_1$
  satisfying $O_{1,t+2} = I_{1,t}$.
  Several wires can be connected to create a longer wire.
  Figure~\ref{livewire} illustrates the inner workings of the wire.
  \item The splitter
  duplicates a signal.
  It is evaluated in $2$ time steps,
  has one input $I_1$, and two outputs $O_1,O_2$
  satisfying $O_{1,t+2}=I_{1,t}$ and $O_{2,t+2} = I_{1,t}$.
  \item The OR gate
    computes the logical disjunction.
  It is evaluated in $4$ time steps,
  has two inputs $I_1,I_2$,
  and one output $O_1$ satisfying $O_{1,t+4} = I_{1,i} \vee I_{2,i}$.
  \item The AND gate (Figure~\ref{andpicture})
  computes the logical conjunction.
  It is evaluated in $4$ time steps,
  has two inputs $I_1,I_2$,
  and one output $O_1$ satisfying $O_{1,t+4} = I_{1,t} \wedge I_{2,t}$.
  \item The NOT gate
  computes the logical negation.
  It is evaluated in $4$ time steps,
  has two inputs (a clock signal is required to activate the gate),
  and one output $O_1$ satisfying $O_{1,t+4} = \neg I_{2,t} \wedge I_{1,t}$.
\end{itemize}

To create a Turing machine, we use one more gadget:
the storage unit (Figure~\ref{supicture}).
Contrary to the previous gadgets,
it does not necessarily erase itself after use,
and can store one bit of information
that can be sent upon request.
Formally, a storage unit has a state $S \in \{0,1\}$,
two inputs $I_1,I_2$,
and one output $O_1$.
The first input is used to modify the current state:
if $I_{1,t}$ is true,
then the storage unit changes its state in four steps.
The second input is used to make the storage unit
broadcast its current state:
$O_{1,t+4} = S \wedge I_{2,t}$.

Note that every gadget has a fixed number of vertices.

\subsection{Functions}

Our basic gadgets compute the basic logical operators.
We show that combining them yields bigger gadgets that can compute any binary
function, with a small restriction:
it is not possible to produce a positive signal out of a negative signal.
For example, our NOT gate needs a clock signal to be activated.
Therefore, we only consider binary functions that map to $0$
all the tuples starting with a $0$.
The proof, technical but straightforward, can be found in Appendix~\ref{appendix:proofEveryFunction}.
\begin{lemma}\label{lemma:everyFunction}
  Let $c \in \mathbb{N}$.
  For every function $f: \{0,1\}^c \rightarrow \{0,1\}$
  mapping to $0$ every tuple whose first component is $0$,
  we can construct a gadget computing $f$
  that is composed of $\calO(2^c)$ basic gadgets,
  and is evaluated in $\calO(c)$ steps.
\end{lemma}

\subsection{Simulating the Turing machine}
We now show how to simulate a Turing machine with a graph following $\overpopulation{2}{1}$.
\begin{lemma}\label{lemma:simulating_t_m}
  Let $T$ be a Turing machine. For every input $u$
  evaluated by $T$ using $C \in \mathbb{N}$ cells of the tape,
  there exists a bounded degree graph $G$ on $\calO(C)$ vertices
  and an initial configuration $c_0$ of $G$
  such that $T$ stops over the input $u$ if and only if updating
  $c_0$ with the overpopulation rule $\overpopulation{2}{1}$
  eventually yields the configuration with only dead vertices
\end{lemma}
\begin{proof}
  We suppose that the Turing machine $T$ has a single final state,
  which can only be accessed after clearing the tape.
  We present the construction of the graph $G$ simulating $T$ through the following steps.
  First, we encode the states of $T$,
  the tape alphabet,
  and the transition function in binary.
  Then, we introduce the notion of blob, the building blocks of $G$,
  and we show that blobs are able to accurately simulate the transition function of $T$.
  Afterwards, we approximate the size of a blob, and finally we define $G$.
  
  \subparagraph{Binary encoding.}
  Let $T_s \in \mathbb{N}$
  be the number of states of $T$,
  and $T_a \in \mathbb{N}$ be the size of its tape alphabet.
  We pick two small integers $s$ and $n$ satisfying $T_s \leq 2^{s-1}$
  and $T_a \leq 2^{n-1}$.
  We encode the states of $T$ as elements of $\{0, 1\}^s$,
  and the alphabet symbols as elements of $\{0, 1\}^n$,
  while respecting the following three conditions: the blank symbol is mapped to
  $0^n$, the final state of $T$ is mapped to $0^s$,
  and all the other states are mapped to strings starting with $1$.
  Then, with respect to these mappings, we modify the transition function of $T$ to:
  \[
  F: \{0, 1\}^s \times \{0, 1\}^n \rightarrow
  \{0, 1\}^s \times \{0, 1\}^s \times \{0, 1\}^n.
  \]
  Instead of using one bit to denote the movement,
  we use $2s$ bits to store the state and signify the movement:
  if the first $s$ bits are zero, the head is moving right;
  if the second $s$ bits are zero, the head is moving left;
  if the first $2s$ bits are zero, the computation ended.
  Moreover, the last $n$ bits of the image of $F$ do not encode the new symbol,
  but the symmetric difference between the previous and the next symbol:
  if the $i$-th bit of the tape symbol goes from $y_i$ to $z_i$,
  then $F$ outputs $d_i = y_i \oplus z_i$ (XOR of these two).

  \subparagraph{Constructing blobs.}
  As we said at the beginning, the graph $G$ is obtained by simulating
  each cell of the tape with a blob,
  which is a gadget storing a tape symbol, and that is able,
  when receiving a signal corresponding to a state of $T$, to
  compute the corresponding result of the transition function.
  The main components of a blob are as follows.  
  \begin{itemize}
  \item
  Memory: $n$ storage units ($s_1$, $s_2$, $\dots$, $s_n$)
  are used to keep in memory a tape symbol $a \in \{0,1\}^n$ of $T$.
  \item
  Receptor: $2s$ inputs ($I_1$, $I_2$, $\dots$, $I_{2s}$)
  are used to receive states $q \in \{0,1\}^s$ of $T$
  either from the left or from the right.
  \item
  Transmitter: $2s$ outputs ($O_1$, $O_2$, $\dots$, $O_{2s}$)
  are used to send states $q \in \{0,1\}^s$ of $T$
  either to the right or to the left.
  \item
  Transition gadget: using Lemma~\ref{lemma:everyFunction},
  we create a gadget computing each of the $2s+n$ output bits of $F$.
  These gadgets are then combined into a bigger gadget $g_F$
  that evaluates them all separately in parallel, and computes
  the transition function $F$.
  Note that $g_F$ is composed of $\calO((n+s)2^{n+s})$
  basic gadgets, and its evaluation time is $\calO(n+s)$.
  \end{itemize}
   
  Blobs are connected in a row to act as a tape:
  for every $1 \leq i \leq s$,
  the output $O_i$ of each blob is connected to
  the input $I_i$ of the blob to its right,
  and the output $O_{s+i}$ of each blob is connected to
  the input $I_{s+i}$ of the blob to its left.
  When receiving a signal,
  the blob transmits the received state
  and the tape symbol stored in memory
  to the transition gadget $g_F$,
  which computes the corresponding transition,
  and then apply its results. 
  We now detail this inner behavior.
  Note that when a gadget is supposed to receive
  simultaneously a set of signals coming from different sources,
  it is always possible to add wires of adapted length
  to ensure that all the signals end up synchronized.

  \subparagraph{Simulating the transition function.}
  In order to simulate the transition function of $T$,
  a blob acts according to the three following steps:
  
  \smallskip
  \noindent
  \textbf{1.} Transmission of the state.
  A blob can receive a state
  either from the left (through inputs $I_1,I_2,\ldots,I_s$)
  or from the right (through inputs $I_{s+1},I_{s+2},\ldots,I_{2s}$),
  but not from both sides at the same time,
  since at every point in time there is at most one active state.
  Therefore, if for every $1 \leq i \leq s$ we denote by $x_i$
  the disjunction of the signals received by $I_i$ and $I_{s+i}$,
  then the resulting tuple $(x_1,x_2, \ldots, x_s)$ is equal to
  the state received as signal (either from the left or the right),
  which can be fed to the gadget $g_F$.
  Formally, the blob connects,
  for all $1 \leq i \leq s$,
  the pair $I_i$, $I_{s+i}$ to an OR gate
  whose output is linked to the input $I_i$ of $g_F$.
  
  \smallskip
  \noindent
 \textbf{2.} Transmission of the tape symbol.
  Since the first component of any state apart from the final state
  is always $1$,
  whenever a blob receives a state,
  the component $x_1$ defined in the previous paragraph has value $1$.
  The tape symbol $(y_1,y_2, \ldots, y_n)$
  currently stored in the blob can be obtained by sending,
  for every $1 \leq i \leq n$,
  a copy of $x_1$ to the input $I_2$ of the storage unit $s_i$,
  causing it to broadcast its stored state $y_i$.
  The tuple can then be fed to the gadget $g_F$.
  Formally, the blob uses $n$ splitters
  to transmit the result of the OR gate
  between $I_1$ and $I_{s+1}$ to
  the input $I_1$ of each storage unit.
  Then, for every $1 \leq i \leq n$,
  the output $O_1$ of the storage unit $s_i$
  is connected to the input $I_{s+i}$ of $g_F$.

  \smallskip
  \noindent
  \textbf{3.} Application of the transition.
  Upon receiving a state and a tape
  symbol, $g_F$ computes the result of the transition
  function, yielding a tuple $(r_1,r_2, \ldots, r_{s+n})$.
  The blob now needs to do two things: send a state to the successor blob,
  and update the element of the tape.
    
  Connecting the output $O_i$ of $g_F$
  to the output $O_i$ of the blob for every $1 \leq i \leq 2s$
  ensures that the state is sent to the correct neighbor:
  the values $(r_{1},r_2, \ldots, r_s)$ are nonzero
  if the head is supposed to move to the right,
  and the outputs $O_1,O_2, \ldots,O_s$ of the blob are connected to the right.
  Conversely, $(r_{s+1},r_{s+2}, \ldots, r_{2s})$ is nonzero
  if the head is supposed to move to the left,
  and the outputs $O_{s+1},O_{s+2}, \ldots,O_{2s}$ of the blob
  are connected to the left.
  
  Finally, connecting the output $O_{2s + i}$ of $g_F$
  to the input $I_1$ of $s_i$ for all $1 \leq i \leq n$
  ensures that the state is correctly updated:
  this sends the signal $d_i$ to the input $I_1$ of the storage unit $s_i$.
  Since $d_i$ is the difference between the current bit and the next,
  the state of $s_i$ will change only if it has to.
  
  \subparagraph{Size of a blob.}
  To prepare the signal for the transition function and to send the signal to another blob,
  only $\calO(n+s)$ basic gadgets and $\calO(n)$ steps are needed.
  As a consequence, the size of a blob is mainly determined by
  the size of the transition gadget $g_F$:
  one blob is composed of $\calO ((n+s)2^{n+s})$ basic gadgets of constant size,
  and evaluating a transition requires $\calO(n+s)$ steps.
  Since $n$ and $s$ are constants (they depend on $T$, and not on the input $u$),
  the blob has constant size.
  Moreover, all the basic gadgets used in the construction have bounded degree,
  so the blob also has bounded degree.
  
  \subparagraph{Constructing $G$.}
  Now that we have blobs that accurately
  simulate the transition function of $T$, constructing the graph
  $G$ simulating the behavior of $T$ over the input $u$ is easy:
  we simply take a row of $C$ blobs
  (remember that $C \in \mathbb{N}$ is the number of tape cells used by $T$ to process $u$).
  Since the size of a blob is constant, $G$ is polynomial in $C$.
  We define the initial configuration of $G$ by setting the
  states of the $|u|$ blobs on the left of the row to the letters of $u$,
  and setting the inputs $I_1$ to $I_s$ of the leftmost blob to
  the signal corresponding to the initial state of $T$ as if it was already in the process.
  As explained earlier, the blobs then evolve by following the run of $T$.
  If the Turing machine stops, then its tape is empty and the final state is sent.
  Since in $G$ the final state is encoded by $0^s$
  and the blank symbol is encoded by $0^n$,
  this results in $G$ reaching the configuration where all
  the vertices are dead.
  Conversely, if $T$ runs forever starting
  from the input $u$, there will always be some live vertices in
  $G$ to transmit the signal corresponding to the state of $T$.
\end{proof}

\begin{proof}[Proof of Theorem~\ref{theorem:overpopulation}]
  By Lemma~\ref{lemma:simulating_t_m}, we can reduce any problem solvable
  by a polynomially bounded Turing machine into $\reach$,
  asking whether the configuration with only dead vertices is reached,
  or into $\average$,
  asking whether the long-run average is strictly above $0$.
\end{proof}

\section{Conclusion}
In this work, we identify two simple update rules for Game of Life. 
We show (in Section~\ref{section:motivation}) that these simple rules can model several well-studied 
dynamics in the literature.
While we show that efficient algorithms exist for the underpopulation rule, the 
computational problems are PSPACE-hard for the overpopulation rule.
An interesting direction for future work would be to consider whether for certain 
special classes of graphs (e.g., grids) efficient algorithms can be obtained for 
the overpopulation rule.

\bibliography{gol}

\newpage
\appendix
\section{More motivation}\label{appendix:motivation}

\subparagraph{Opinions.}
In order to sell their products,
some companies hire well-known figures to advertise for it,
hoping that other people will follow the celebrity and buy the product.
 
The paper~\cite{public_opinion_formation}
studies how public opinions are formed by \emph{influentials},
well connected individuals that expose others to new ideas or opinions.
They examine the ``influentials hypothesis''
that a small number of influentials can influence the whole network.
By computer simulations of interpersonal influence processes,
they found that the opinion
either get stuck at the influentials,
or is propagated by the network fairly quickly.
In most cases, the spread of opinion does not depends on the connectivity of influentials,
but it is supported by easily convincible individuals.
Then, opinion needs to
get a critical mass of supporters,
otherwise it spreads only in the close proximity of starting vertex.
 
In the model, one individual starts with a fixed opinion $1$,
and the others start with opinion $0$.
Each individual has a willingness to be convinced $\psi \in [0,1]$,
and changes its opinion to $1$ as soon as a fraction of its
neighbors greater than $\psi$ has opinion $1$.
In our setting, this can be expressed as follows:
for every vertex $i$ having a willingness to be conviced $\psi_i$,
and $b_i$ live neighbors,
\[
\begin{array}{ll}
\phi_0(m,i) = \left\{
\begin{array}{lll}
0 \textup{ if } b_i < \psi_i;\\
1 \textup{ if } b_i \geq \psi_i.
\end{array}
\right. &
\phi_1(m,i) = 1
\end{array}
\]
We can view that process as an underpopulation process.
The only difference is that in our process every threshold is the same.

\section{Underpopulation rule}

\subsection{Definitions and Examples}\label{appendix:examples}
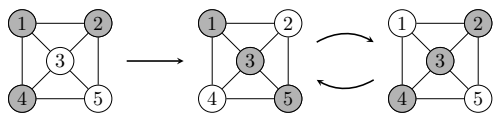
\begin{figure}
  \begin{center}
    \scalebox{0.5}{
    \begin{tikzpicture}
      \foreach \c in {1,6,11}{
        \pgfmathtruncatemacro{\ste}{\c}
        \pgfmathtruncatemacro{\ma}{0}\
        \coordinate (c\c1) at (\ste,\ma+2);
        \node[draw,circle,fill=white,scale=1.5,inner sep=2]
          (n\c1) at (c\c1) {4};
        \coordinate (c\c2) at (\ste,\ma+4);
        \node[draw,circle,fill=white,scale=1.5,inner sep=2]
          (n\c2) at (c\c2) {1};
      }
      \foreach \c in {3,8,13}{
        \pgfmathtruncatemacro{\ste}{\c}
        \pgfmathtruncatemacro{\ma}{0}\
        \coordinate (c\c1) at (\ste,\ma+2);
        \node[draw,circle,fill=white,scale=1.5,inner sep=2]
          (n\c1) at (c\c1) {5};
        \coordinate (c\c2) at (\ste,\ma+4);
        \node[draw,circle,fill=white,scale=1.5,inner sep=2]
          (n\c2) at (c\c2) {2};
      }
      \foreach \c in {2,7,12}{
        \pgfmathtruncatemacro{\ste}{\c}
        \pgfmathtruncatemacro{\ma}{3}
        \pgfmathtruncatemacro{\height}{(\ma)}
        \pgfmathtruncatemacro{\ste}{\c}
        \coordinate (c\c1) at (\ste,\height);
        \node[draw,circle,fill=white,scale=1.5,inner sep=2]
          (n\c1) at (c\c1) {3};
      }
      \foreach \c in {1,6,11}{
        \pgfmathtruncatemacro{\nx}{\c+1}
        \pgfmathtruncatemacro{\nnx}{\c+2}
          \draw (n\c1) -- (n\c2);
          \draw (n\c1) -- (n\nx1);
          \draw (n\c2) -- (n\nx1);
          \draw (n\c1) -- (n\nnx1);
          \draw (n\c2) -- (n\nnx2);
          \draw (n\nx1) -- (n\nnx1);
          \draw (n\nx1) -- (n\nnx2);
          \draw (n\nnx1) -- (n\nnx2);
         
      }
      \node[draw,circle,fill=black!30,scale=1.5,inner sep=2]
       (n11) at (c11) {4};
      \node[draw,circle,fill=black!30,scale=1.5,inner sep=2]
       (n12) at (c12) {1};
      \node[draw,circle,fill=black!30,scale=1.5,inner sep=2]
       (n32) at (c32) {2};
      \node[draw,circle,fill=black!30,scale=1.5,inner sep=2]
       (n62) at (c62) {1};
      \node[draw,circle,fill=black!30,scale=1.5,inner sep=2]
       (n71) at (c71) {3};
      \node[draw,circle,fill=black!30,scale=1.5,inner sep=2]
       (n81) at (c81) {5};
      \node[draw,circle,fill=black!30,scale=1.5,inner sep=2]
       (n111) at (c111) {4};
      \node[draw,circle,fill=black!30,scale=1.5,inner sep=2]
       (n121) at (c121) {3};
      \node[draw,circle,fill=black!30,scale=1.5,inner sep=2]
       (n132) at (c132) {2};

      \draw[-{stealth},very thick]
       ($(c21) + (1.75,0)$) to ($(c71)-(1.75,0)$);
      \draw[-{stealth},very thick]
       ($(c71) + (1.75,0.5)$) to [bend left] ($(c121)-(1.75,-0.5)$);
      \draw[-{stealth},very thick]
       ($(c121) - (1.75,0.5)$) to [bend left] ($(c71)+(1.75,-0.5)$);
    \end{tikzpicture}
    }
  \end{center}
  \caption{Evolution of a graph
  under the underpopulation rule
  $\underpopulation{2}{2}$.
  Live vertices are gray.}\label{underpopulation_example_picture_app}
\end{figure}

We define some technical notions that are used in our proofs.
We then illustrate these notions
using the setting presented by Figure \ref{underpopulation_example_picture_app}.

\subparagraph{Counting prefixes and suffixes.}
Given an integer
$m \in \mathbb{N}$
and a word $y \in \{0,1,\wildcard \}^*$
of size $n \leq m$,
on top of counting factors of histories that match the regular expression $y$,
we are also interested in the number
of vertices $s \in V$ whose history
$\word{s}[0,m]$
admits a prefix, resp. suffix,
matching $y$.
\begin{align*} 
&\numbu{\vdash y}_m
= \big{|}
\{
s \in V |
\word{s}[0,n] = y
\}
\big{|};\\
&\numbu{y \dashv}_m
= \big{|}
\{
s \in V |
\word{s}[m-n,m] = y
\}
\big{|}.
\end{align*}

\subparagraph{Counting synchronised factors.}
We introduce a way of counting
the number of synchronised occurrences of a pair
of factors $y,z \in \{0,1,\wildcard \}^*$
of same size $n \leq m$
in histories $\word{s}[0,m]$ and $\word{t}[0,m]$
corresponding to neighbour vertices $(s,t) \in E$.
\begin{align*}
&\numbu{y,z}_m
= \\
& \rlap{$
\ \big{|}
\{
(s,t,i) \in E \times \mathbb{N} |
i+n \leq m,
\word{s}[i,i+n] = y,
\word{t}[i,i+n] = z
\}
\big{|};$}\\
&\numbu{\vdash y,\vdash z}_m
=
\big{|}
\{
(s,t) \in E |
\word{s}[0,n] = y,\word{t}[0,n] = z,
\}
\big{|};\\
&\numbu{y \dashv,z \dashv}_m
=
\big{|}
\{
(s,t) \in E |
\word{s}[m-n,m] = y,
\word{t}[m-n,m] = z
\}
\big{|}.
\end{align*}

\subparagraph{Examples.}
In the setting 
depicted in Figure \ref{underpopulation_example_picture_app},
the five histories are as follows:
\[
\tau_1 = 1(10)^{\omega} \ \ \
\tau_2 = (10)^{\omega} \ \ \
\tau_3 = 01^{\omega} \ \ \
\tau_4 = (10)^{\omega} \ \ \
\tau_5 = (01)^{\omega}
\]
Let us focus on the prefixes of size $6$
of the histories:
\[
\left.
\begin{array}{l}
\tau_1[0,6] = 110101\\[2pt]
\tau_4[0,6] = 101010
\end{array}
\right.
\
\tau_3[0,6] = 011111
\
\left.
\begin{array}{l}
\tau_2[0,6] = 101010\\[2pt]
\tau_5[0,6] = 010101
\end{array}
\right.
\]
We observe that the number of occurrences 
of diverse factors $y \in \{0,1,\wildcard \}^*$
of length $3$ is clearly unbalanced:
\begin{align*}
&\numbu{1 \wildcard 1}_6 = 11&
&\numbu{0 \wildcard 0}_6 = 7&
&\numbu{1 \wildcard 0}_6 = 1&
&\numbu{0 \wildcard 1}_6 = 1&
\end{align*}
Note that the divergence increases if we consider
longer prefixes:
both $1 \wildcard 1$ and $0 \wildcard 0$
appear infinitely often,
while $1 \wildcard 0$ and $0 \wildcard 1$
do not appear more than a single time each.
Our key lemma states that such a disparity
always happens:
we show that the number of occurrences
of $1 \wildcard 0$ and $0 \wildcard 1$ is always bounded
(which implies Proposition \ref{theorem:underpopulation_cycle}),
and that the bound is polynomial with respect to
$G$, $i_0$ and $i_1$
(which implies Proposition \ref{proposition:underpopulation}).

Let us now focus on the prefixes of size $4$
of the histories:
\[
\left.
\begin{array}{l}
\tau_1[0,4] = 1101\\[2pt]
\tau_4[0,4] = 1010
\end{array}
\right.
\
\tau_3[0,4] = 0111
\
\left.
\begin{array}{l}
\tau_2[0,4] = 1010\\[2pt]
\tau_5[0,4] = 0101
\end{array}
\right.
\]
Note that $\numbu{1 0 1}_4 = 4$ and $\numbu{0 1 0}_4 = 3$:
the factor $010$ appear once
in the history of the vertices $2,4$ and $5$,
and $101$ also appears in the history of $1$.
Since each of these vertices has exactly three neighbours,
the number of synchronised pairs (see Part \ref{appendix:lemma} of the appendix for the definition)
corresponding to both factors is:
\begin{align*}
&\numbu{1 0 1,\wildcard \wildcard \wildcard}_4 = 4 \cdot 3 =12,&
&\numbu{0 1 0,\wildcard \wildcard \wildcard}_4 = 3 \cdot 3 =9.&
\intertext{
Let us focus on the second symbol of the second parameter.
All of the factors synchronised with $101$ have a $1$ in the middle,
while for each of the three copies of the factor $010$,
only the synchronised factor corresponding to the vertex $3$
has a $1$ in the middle, and the other two have a $0$:}
&\numbu{1 0 1,\wildcard 1 \wildcard}_4 = 12,&
&\numbu{0 1 0,\wildcard 1 \wildcard}_4 = 3,&\\
&\numbu{1 0 1,\wildcard 0 \wildcard}_4 = 0,&
&\numbu{0 1 0,\wildcard 0 \wildcard}_4 = 6.&
\end{align*}
These divergences are expected
as the update rule is $\underpopulation{2}{2}$:
for each occurrence of the factor $101$,
since the last symbol is $1$,
at least two neighbors have to be alive at the step corresponding to the second symbol,
i.e.,
$\numbu{1 0 1,\wildcard 1 \wildcard}_4 \geq 2 \cdot \numbu{1 0 1}_4$.
Conversely,
for each occurrence of the factor $010$,
since the third symbol is $0$,
at most one neighbor is alive at the step corresponding to the second symbol,
i.e.,
$\numbu{0 1 0,\wildcard 1 \wildcard}_4 \leq 1 \cdot \numbu{0 1 0}_4$.

\subsection{Proof of Lemma \ref{lemma:bound_change}}\label{appendix:lemma}

We present the formal proof of our key Lemma.

\keyLemma*

\begin{proof}[Proof of Lemma \ref{lemma:bound_change}]
Let us fix some integer $m \geq 3$.
We obtain the proof by combining
the following equations resulting from simple observations
on the number of diverse factors in the prefixes
$\word{s}[1,m]$ of the histories of $G$.
In order to lighten the notation,
for every $y,z \in \{0,1,\wildcard \}^*$
we write 
$\numbu{y}$ instead of $\numbu{y}_m$
and 
$\numbu{y,z}$ instead of $\numbu{y,z}_m$.
\begin{align*} 
&\textup{1. }\numbu{100,\wildcard 1 \wildcard} \leq
(i_0-1)\numbu{100}; & 
&\textup{2. }\numbu{001,\wildcard 1 \wildcard} \geq
i_0\numbu{001};\\
&\textup{3. }
\numbu{110,\wildcard 1 \wildcard} \leq
(i_1-1)\numbu{110}; &
&\textup{4. }
\numbu{011,\wildcard 1 \wildcard} \geq
i_1\numbu{011};\\
&\textup{5. }
|\numbu{100} - \numbu{001}| \leq |V|; &
&\textup{6. }
|\numbu{011} - \numbu{110}| \leq |V|;\\
&\rlap{$\textup{7. }
|\numbu{001,\wildcard 1 \wildcard} +
\numbu{011,\wildcard 1 \wildcard} -
\numbu{100,\wildcard 1 \wildcard}-
\numbu{110,\wildcard 1 \wildcard}| \leq |E|;$}
\end{align*}
We begin by detailing how these seven equations are obtained.
The first four equations are direct
rephrasing of the underpopulation rule
$\underpopulation{i_0}{i_1}$.
\begin{itemize}
\item
A dead vertex becomes alive if and only if it has
at least $i_0$ live neighbours.
Therefore, whenever the factor $100$ appears in a history,
at most $i_0-1$ factors of the form $\wildcard 1 \wildcard$
are synchronised with it,
since otherwise the third element of the factor would be
a $1$ instead of a $0$.
Conversely, whenever the factor $101$ appears,
it needs to be synchronised with at least $i_0$
factors of the form $\wildcard 1 \wildcard$
to guarantee that the third element is a $1$.
Formally,
\begin{align}
\numbu{1 0 0,\wildcard 1 \wildcard} 
&\leq
(i_0-1) \numbu{1 0 0};
\label{firsta}\\
\numbu{0 0 1,\wildcard 1 \wildcard} 
&\geq
i_0 \numbu{0 0 1};
\label{firstb}
\end{align}
\item
A live vertex stays alive if and only if it has
at least $i_1$ live neighbours:
\begin{align}
\numbu{1 1 0,\wildcard 1 \wildcard} 
&\leq
(i_1-1) \numbu{1 1 0};
\label{seconda}\\
\numbu{0 1 1,\wildcard 1 \wildcard} 
&\geq
i_1 \numbu{0 1 1}.
\label{secondb}
\end{align}
\end{itemize}
Each of the last three equations is obtained by
counting a given factor in two different ways.
First, the occurrences of the factors $00$ and $11$
can be counted either by differentiating the previous letter,
or the next one.
\[
\begin{array}{rcl}
\numbu{\vdash 00} +
\numbu{000} +
\numbu{100}
& = &
\numbu{00}\\
& = &
\numbu{000} +
\numbu{001} +
\numbu{00\dashv};\\
\numbu{\vdash 11} +
\numbu{111} +
\numbu{011}
& = &
\numbu{11}\\
& = &
\numbu{110} +
\numbu{111} +
\numbu{11\dashv}.
\end{array}
\]
Moving the terms around,
cancelling
$\numbu{000}$
and $\numbu{111}$,
and using the fact that
the terms containing
the $\vdash$ and $\dashv$ symbols
are between $0$ and $|V|$
yields our next two equations.
\begin{align}
\label{equ:fac00}
|\numbu{100}-
\numbu{001}|
& =
|\numbu{\vdash 00}-
\numbu{00\dashv}|
\leq |V|.\\
\label{equ:fac11}
|\numbu{011}-
\numbu{110}|
& =
|\numbu{\vdash 11}-
\numbu{11\dashv}|
\leq |V|.
\end{align}
Finally, 
for every $y,z \in \{0,1,\wildcard \}^*$,
the equation $\numbu{y,z} = \numbu{z,y}$ holds
as the definition is symmetrical.
Therefore, in particular,
$
\numbu{\wildcard 1,1 \wildcard} =
\numbu{1 \wildcard,\wildcard 1}
$.
Let us expand this equality
by considering the possible previous symbols on the left side
and the possible next symbols on the right side.
\[
\begin{array}{lll}
\numbu{\wildcard 1,1 \wildcard}
& = &
\numbu{\vdash \wildcard 1,\vdash 1 \wildcard} +
\numbu{0 \wildcard 1,\wildcard 1 \wildcard}  +
\numbu{1 \wildcard 1,\wildcard 1 \wildcard}\\
& = &
\numbu{\vdash \wildcard 1,\vdash 1 \wildcard} +
\numbu{0 0 1,\wildcard 1 \wildcard}  +
\numbu{0 1 1,\wildcard 1 \wildcard}  +
\numbu{1 \wildcard 1,\wildcard 1 \wildcard};\\
\numbu{1 \wildcard,\wildcard 1} & = &
\numbu{1 \wildcard  \dashv,\wildcard 1 \dashv} + 
\numbu{1 \wildcard 0,\wildcard 1 \wildcard} + 
\numbu{1 \wildcard 1,\wildcard 1 \wildcard}\\
& = &
\numbu{1 \wildcard  \dashv,\wildcard 1 \dashv} + 
\numbu{1 0 0,\wildcard 1 \wildcard} + 
\numbu{1 1 0,\wildcard 1 \wildcard} + 
\numbu{1 \wildcard 1,\wildcard 1 \wildcard}.
\end{array}
\]
Moving the terms around,
cancelling
$\numbu{1 \wildcard 1,\wildcard 1 \wildcard}$,
and using the fact that
the terms containing
the $\vdash$ and $\dashv$ symbols
are between $0$ and $|E|$
yields our final equation.
\begin{equation}\label{equ:final}
|\numbu{001,\wildcard 1 \wildcard} +
\numbu{011,\wildcard 1 \wildcard} -
\numbu{100,\wildcard 1 \wildcard}-
\numbu{110,\wildcard 1 }|
\leq
|E|.
\end{equation}

Now that the equations are proved,
we use them to demonstrate the statement.
First, combining Equations
\ref{equ:fac00} and \ref{equ:fac11}
allows us to bound
$\numbu{0\wildcard 1}$ with
$\numbu{1\wildcard 0} + 2|V|$:
\[
\begin{array}{lll}
\numbu{0\wildcard 1}
& = &
\numbu{001}+
\numbu{011}\\
& \leq &
\numbu{100}+|V|+
\numbu{110}+|V|\\
& \leq &
\numbu{1 \wildcard 0}
+ 2|V|.
\end{array}
\]
As a consequence,
in order to conclude the demonstration,
we only need to bound $\numbu{1\wildcard 0}$ with
$|E| + (i_0+i_1)|V|$.
This is done as follows.
By applying equations \ref{firsta}, \ref{firstb},
\ref{seconda} and \ref{secondb} to Equation
\ref{equ:final},
and rearranging the terms,
we get
\[
\numbu{1 0 0} + \numbu{1 1 0}
\leq
|E|+
i_0(
\numbu{1 0 0} -
% \numbu{\vdash 0 1} -
\numbu{0 0 1}
)+
i_1(
\numbu{1 1 0} -
\numbu{0 1 1}
).
\]
This implies the desired inequality,
as
$\numbu{1 0 0} + \numbu{1 1 0} = \numbu{1 \wildcard 0}$,
and
the content of both parentheses
can be over-approximated with $|V|$
through the use of Equations \ref{equ:fac00}
and \ref{equ:fac11}.
\end{proof}

\section{Overpopulation rule}
\subsection{Bounded space algorithm}\label{appendix:overpopulationAlgorithm}

We show PSPACE algorithms solving
the problems $\reach$ and $\average$
for any overpopulation rule.

\begin{lemma}\label{lemma:bounded_space_algorithm}
  For every $i_0,i_1 \in \mathbb{N}$,
  the problems $\reach$ and $\average$ for $\overpopulation{i_0}{i_1}$ are in PSPACE.
\end{lemma}
\begin{proof}
  The reachability problem $\reach$ for a graph with $n$ vertices can be solved by a simple simulation:
  We simulate the first $2^n$ steps of the process starting from the initial configuration
  (remark that counting up to $2^n$ in binary requires only $n$ bits).
  If we see the target configuration along the way, we answer yes.
  Otherwise, we answer no:
  since a graph with $n$ vertices admits $2^n$ distinct configurations,
  we know for sure that we already completed a cycle of configurations,
  and no new configuration will be reached.

  To solve $\average$ for a graph with $n$ vertices,
  we first simulate the process to compute the configuration $c$
  reached in $2^n$ steps from the initial configuration.
  For the same reasons as before,
  we know that $c$ is part of a cycle of configurations.
  Then, we simulate the process from $c$ until we see $c$ again,
  summing the number of live vertices encountered along the way,
  and keeping track of the length of the cycle.
  This allows us to compute the average of live vertices in the cycle reached from the initial configuration.
%  Both algorithms have to remember only a constant number of configurations, so they are space-bounded.
\end{proof}

\begin{figure}
  \begin{minipage}{0.45\textwidth}
    
            %\begin{figure}
            \begin{center}
            \scalebox{0.5}{
            
\begin{tikzpicture}
\node (Wireanno1) at (-1.3, 1.2) {  \LARGE $I_1$};
\node[scale=5,color=black!30!white] (Wireanno1) at (-0.4, 1.2) {  \LARGE \textbraceleft };
\node (Wireanno3) at (9.45, 1.2) {  \LARGE $O_1$};
\node[scale=5,color=black!30!white] (Wireanno3) at (8.399999999999999, 1.2) {  \LARGE \textbraceright };
\node[draw,circle,fill=white] (0Wire0) at (0, 0.0) {};
\node[draw,circle,fill=white] (0Wire1) at (0, 0.8) {};
\node[draw,circle,fill=white] (0Wire2) at (0, 1.6) {};
\node[draw,circle,fill=white] (0Wire3) at (0, 2.4000000000000004) {};
\node[draw,circle,fill=white] (1Wire0) at (1, 0.8) {};
\node[draw,circle,fill=white] (1Wire1) at (1, 1.6) {};
\node[draw,circle,fill=white] (2Wire0) at (2, 0.0) {};
\node[draw,circle,fill=white] (2Wire1) at (2, 0.8) {};
\node[draw,circle,fill=white] (2Wire2) at (2, 1.6) {};
\node[draw,circle,fill=white] (2Wire3) at (2, 2.4000000000000004) {};
\node[draw,circle,fill=white] (3Wire0) at (3, 0.8) {};
\node[draw,circle,fill=white] (3Wire1) at (3, 1.6) {};
\node[draw,circle,fill=white] (4Wire0) at (4, 0.0) {};
\node[draw,circle,fill=white] (4Wire1) at (4, 0.8) {};
\node[draw,circle,fill=white] (4Wire2) at (4, 1.6) {};
\node[draw,circle,fill=white] (4Wire3) at (4, 2.4000000000000004) {};
\node[draw,circle,fill=white] (5Wire0) at (5, 0.8) {};
\node[draw,circle,fill=white] (5Wire1) at (5, 1.6) {};
\node[draw,circle,fill=white] (6Wire0) at (6, 0.0) {};
\node[draw,circle,fill=white] (6Wire1) at (6, 0.8) {};
\node[draw,circle,fill=white] (6Wire2) at (6, 1.6) {};
\node[draw,circle,fill=white] (6Wire3) at (6, 2.4000000000000004) {};
\node[draw,circle,fill=white] (7Wire0) at (7, 0.8) {};
\node[draw,circle,fill=white] (7Wire1) at (7, 1.6) {};
\node[draw,circle,fill=white] (8Wire0) at (8, 0.0) {};
\node[draw,circle,fill=white] (8Wire1) at (8, 0.8) {};
\node[draw,circle,fill=white] (8Wire2) at (8, 1.6) {};
\node[draw,circle,fill=white] (8Wire3) at (8, 2.4000000000000004) {};
\draw (0Wire0) -- (0Wire0);
\draw (0Wire0) -- (0Wire0);
\draw (0Wire0) -- (1Wire0);
\draw (0Wire1) -- (0Wire1);
\draw (0Wire1) -- (0Wire1);
\draw (0Wire1) -- (1Wire0);
\draw (0Wire2) -- (0Wire2);
\draw (0Wire2) -- (0Wire2);
\draw (0Wire2) -- (1Wire1);
\draw (0Wire3) -- (0Wire3);
\draw (0Wire3) -- (0Wire3);
\draw (0Wire3) -- (1Wire1);
\draw (1Wire0) -- (0Wire0);
\draw (1Wire0) -- (0Wire1);
\draw (1Wire0) -- (2Wire0);
\draw (1Wire0) -- (2Wire1);
\draw (1Wire0) -- (2Wire2);
\draw (1Wire0) -- (2Wire3);
\draw (1Wire1) -- (0Wire2);
\draw (1Wire1) -- (0Wire3);
\draw (1Wire1) -- (2Wire0);
\draw (1Wire1) -- (2Wire1);
\draw (1Wire1) -- (2Wire2);
\draw (1Wire1) -- (2Wire3);
\draw (2Wire0) -- (1Wire0);
\draw (2Wire0) -- (1Wire1);
\draw (2Wire0) -- (3Wire0);
\draw (2Wire1) -- (1Wire0);
\draw (2Wire1) -- (1Wire1);
\draw (2Wire1) -- (3Wire0);
\draw (2Wire2) -- (1Wire0);
\draw (2Wire2) -- (1Wire1);
\draw (2Wire2) -- (3Wire1);
\draw (2Wire3) -- (1Wire0);
\draw (2Wire3) -- (1Wire1);
\draw (2Wire3) -- (3Wire1);
\draw (3Wire0) -- (2Wire0);
\draw (3Wire0) -- (2Wire1);
\draw (3Wire0) -- (4Wire0);
\draw (3Wire0) -- (4Wire1);
\draw (3Wire0) -- (4Wire2);
\draw (3Wire0) -- (4Wire3);
\draw (3Wire1) -- (2Wire2);
\draw (3Wire1) -- (2Wire3);
\draw (3Wire1) -- (4Wire0);
\draw (3Wire1) -- (4Wire1);
\draw (3Wire1) -- (4Wire2);
\draw (3Wire1) -- (4Wire3);
\draw (4Wire0) -- (3Wire0);
\draw (4Wire0) -- (3Wire1);
\draw (4Wire0) -- (5Wire0);
\draw (4Wire1) -- (3Wire0);
\draw (4Wire1) -- (3Wire1);
\draw (4Wire1) -- (5Wire0);
\draw (4Wire2) -- (3Wire0);
\draw (4Wire2) -- (3Wire1);
\draw (4Wire2) -- (5Wire1);
\draw (4Wire3) -- (3Wire0);
\draw (4Wire3) -- (3Wire1);
\draw (4Wire3) -- (5Wire1);
\draw (5Wire0) -- (4Wire0);
\draw (5Wire0) -- (4Wire1);
\draw (5Wire0) -- (6Wire0);
\draw (5Wire0) -- (6Wire1);
\draw (5Wire0) -- (6Wire2);
\draw (5Wire0) -- (6Wire3);
\draw (5Wire1) -- (4Wire2);
\draw (5Wire1) -- (4Wire3);
\draw (5Wire1) -- (6Wire0);
\draw (5Wire1) -- (6Wire1);
\draw (5Wire1) -- (6Wire2);
\draw (5Wire1) -- (6Wire3);
\draw (6Wire0) -- (5Wire0);
\draw (6Wire0) -- (5Wire1);
\draw (6Wire0) -- (7Wire0);
\draw (6Wire1) -- (5Wire0);
\draw (6Wire1) -- (5Wire1);
\draw (6Wire1) -- (7Wire0);
\draw (6Wire2) -- (5Wire0);
\draw (6Wire2) -- (5Wire1);
\draw (6Wire2) -- (7Wire1);
\draw (6Wire3) -- (5Wire0);
\draw (6Wire3) -- (5Wire1);
\draw (6Wire3) -- (7Wire1);
\draw (7Wire0) -- (6Wire0);
\draw (7Wire0) -- (6Wire1);
\draw (7Wire0) -- (8Wire0);
\draw (7Wire0) -- (8Wire1);
\draw (7Wire0) -- (8Wire2);
\draw (7Wire0) -- (8Wire3);
\draw (7Wire1) -- (6Wire2);
\draw (7Wire1) -- (6Wire3);
\draw (7Wire1) -- (8Wire0);
\draw (7Wire1) -- (8Wire1);
\draw (7Wire1) -- (8Wire2);
\draw (7Wire1) -- (8Wire3);
\draw (8Wire0) -- (7Wire0);
\draw (8Wire0) -- (7Wire1);
\draw (8Wire1) -- (7Wire0);
\draw (8Wire1) -- (7Wire1);
\draw (8Wire2) -- (7Wire0);
\draw (8Wire2) -- (7Wire1);
\draw (8Wire3) -- (7Wire0);
\draw (8Wire3) -- (7Wire1);
\end{tikzpicture}

            }
            \end{center}
            \caption{ Four connected wires used to transmit the signal.}\label{wirepicture}
            %\end{figure}

  \end{minipage}
  \hspace{0.1\textwidth}
  \begin{minipage}{0.45\textwidth}
    
            %\begin{figure}
            \begin{center}
            \scalebox{0.5}{
            
\begin{tikzpicture}
\node (Splitteranno1) at (-1.3, 4.2) {  \LARGE $I_1$};
\node[scale=5,color=black!30!white] (Splitteranno1) at (-0.4, 4.2) {  \LARGE \textbraceleft };
\node (Splitteranno2) at (3.48, 6.23) {  \LARGE $O_1$};
\node[scale=5,color=black!30!white] (Splitteranno2) at (2.4299999999999997, 6.23) {  \LARGE \textbraceright };
\node (Splitteranno3) at (3.48, 2.23) {  \LARGE $O_2$};
\node[scale=5,color=black!30!white] (Splitteranno3) at (2.4299999999999997, 2.23) {  \LARGE \textbraceright };
\node[draw,circle,fill=white] (0Splitterupwiresp0) at (2, 5.0) {};
\node[draw,circle,fill=white] (0Splitterupwiresp1) at (2, 5.8) {};
\node[draw,circle,fill=white] (0Splitterupwiresp2) at (2, 6.6) {};
\node[draw,circle,fill=white] (0Splitterupwiresp3) at (2, 7.4) {};
\node[draw,circle,fill=white] (0Splitterdownwieresp0) at (2, 1.0) {};
\node[draw,circle,fill=white] (0Splitterdownwieresp1) at (2, 1.8) {};
\node[draw,circle,fill=white] (0Splitterdownwieresp2) at (2, 2.6) {};
\node[draw,circle,fill=white] (0Splitterdownwieresp3) at (2, 3.4000000000000004) {};
\node[draw,circle,fill=white] (0Splitterw0) at (0, 3.0) {};
\node[draw,circle,fill=white] (0Splitterw1) at (0, 3.8) {};
\node[draw,circle,fill=white] (0Splitterw2) at (0, 4.6) {};
\node[draw,circle,fill=white] (0Splitterw3) at (0, 5.4) {};
\node[draw,circle,fill=white] (1Splitterw0) at (1, 3.8) {};
\node[draw,circle,fill=white] (1Splitterw1) at (1, 4.6) {};
\draw (0Splitterupwiresp0) -- (0Splitterupwiresp0);
\draw (0Splitterupwiresp0) -- (0Splitterupwiresp0);
\draw (0Splitterupwiresp0) -- (1Splitterw0);
\draw (0Splitterupwiresp0) -- (1Splitterw1);
\draw (0Splitterupwiresp1) -- (0Splitterupwiresp1);
\draw (0Splitterupwiresp1) -- (0Splitterupwiresp1);
\draw (0Splitterupwiresp1) -- (1Splitterw0);
\draw (0Splitterupwiresp1) -- (1Splitterw1);
\draw (0Splitterupwiresp2) -- (0Splitterupwiresp2);
\draw (0Splitterupwiresp2) -- (0Splitterupwiresp2);
\draw (0Splitterupwiresp2) -- (1Splitterw0);
\draw (0Splitterupwiresp2) -- (1Splitterw1);
\draw (0Splitterupwiresp3) -- (0Splitterupwiresp3);
\draw (0Splitterupwiresp3) -- (0Splitterupwiresp3);
\draw (0Splitterupwiresp3) -- (1Splitterw0);
\draw (0Splitterupwiresp3) -- (1Splitterw1);
\draw (0Splitterdownwieresp0) -- (0Splitterdownwieresp0);
\draw (0Splitterdownwieresp0) -- (0Splitterdownwieresp0);
\draw (0Splitterdownwieresp0) -- (1Splitterw0);
\draw (0Splitterdownwieresp0) -- (1Splitterw1);
\draw (0Splitterdownwieresp1) -- (0Splitterdownwieresp1);
\draw (0Splitterdownwieresp1) -- (0Splitterdownwieresp1);
\draw (0Splitterdownwieresp1) -- (1Splitterw0);
\draw (0Splitterdownwieresp1) -- (1Splitterw1);
\draw (0Splitterdownwieresp2) -- (0Splitterdownwieresp2);
\draw (0Splitterdownwieresp2) -- (0Splitterdownwieresp2);
\draw (0Splitterdownwieresp2) -- (1Splitterw0);
\draw (0Splitterdownwieresp2) -- (1Splitterw1);
\draw (0Splitterdownwieresp3) -- (0Splitterdownwieresp3);
\draw (0Splitterdownwieresp3) -- (0Splitterdownwieresp3);
\draw (0Splitterdownwieresp3) -- (1Splitterw0);
\draw (0Splitterdownwieresp3) -- (1Splitterw1);
\draw (0Splitterw0) -- (0Splitterw0);
\draw (0Splitterw0) -- (0Splitterw0);
\draw (0Splitterw0) -- (1Splitterw0);
\draw (0Splitterw1) -- (0Splitterw1);
\draw (0Splitterw1) -- (0Splitterw1);
\draw (0Splitterw1) -- (1Splitterw0);
\draw (0Splitterw2) -- (0Splitterw2);
\draw (0Splitterw2) -- (0Splitterw2);
\draw (0Splitterw2) -- (1Splitterw1);
\draw (0Splitterw3) -- (0Splitterw3);
\draw (0Splitterw3) -- (0Splitterw3);
\draw (0Splitterw3) -- (1Splitterw1);
\draw (1Splitterw0) -- (0Splitterw0);
\draw (1Splitterw0) -- (0Splitterw1);
\draw (1Splitterw0) -- (0Splitterupwiresp0);
\draw (1Splitterw0) -- (0Splitterupwiresp1);
\draw (1Splitterw0) -- (0Splitterupwiresp2);
\draw (1Splitterw0) -- (0Splitterupwiresp3);
\draw (1Splitterw0) -- (0Splitterdownwieresp0);
\draw (1Splitterw0) -- (0Splitterdownwieresp1);
\draw (1Splitterw0) -- (0Splitterdownwieresp2);
\draw (1Splitterw0) -- (0Splitterdownwieresp3);
\draw (1Splitterw1) -- (0Splitterw2);
\draw (1Splitterw1) -- (0Splitterw3);
\draw (1Splitterw1) -- (0Splitterupwiresp0);
\draw (1Splitterw1) -- (0Splitterupwiresp1);
\draw (1Splitterw1) -- (0Splitterupwiresp2);
\draw (1Splitterw1) -- (0Splitterupwiresp3);
\draw (1Splitterw1) -- (0Splitterdownwieresp0);
\draw (1Splitterw1) -- (0Splitterdownwieresp1);
\draw (1Splitterw1) -- (0Splitterdownwieresp2);
\draw (1Splitterw1) -- (0Splitterdownwieresp3);
\end{tikzpicture}

            }
            \end{center}
            \caption{ Splitter, a gadget used to divide the signal.}\label{splitterpicture}
            %\end{figure}

  \end{minipage}
\end{figure}

\begin{figure}
  \begin{minipage}{0.45\textwidth}
    
            %\begin{figure}
            \begin{center}
            \scalebox{0.5}{
            
\begin{tikzpicture}
\node (Oranno1) at (-1.3, 5.2) {  \LARGE $I_1$};
\node[scale=5,color=black!30!white] (Oranno1) at (-0.4, 5.2) {  \LARGE \textbraceleft };
\node (Oranno2) at (-1.3, 1.2) {  \LARGE $I_2$};
\node[scale=5,color=black!30!white] (Oranno2) at (-0.4, 1.2) {  \LARGE \textbraceleft };
\node (Oranno3) at (5.45, 3.2) {  \LARGE $O_1$};
\node[scale=5,color=black!30!white] (Oranno3) at (4.4, 3.2) {  \LARGE \textbraceright };
\node[draw,circle,fill=white] (0Orcror0) at (0, 0.0) {};
\node[draw,circle,fill=white] (0Orcror1) at (0, 0.8) {};
\node[draw,circle,fill=white] (0Orcror2) at (0, 1.6) {};
\node[draw,circle,fill=white] (0Orcror3) at (0, 2.4000000000000004) {};
\node[draw,circle,fill=white] (1Orcror0) at (1, 0.8) {};
\node[draw,circle,fill=white] (1Orcror1) at (1, 1.6) {};
\node[draw,circle,fill=white] (2Orcror0) at (2, 0.0) {};
\node[draw,circle,fill=white] (2Orcror1) at (2, 0.8) {};
\node[draw,circle,fill=white] (2Orcror2) at (2, 1.6) {};
\node[draw,circle,fill=white] (2Orcror3) at (2, 2.4000000000000004) {};
\node[draw,circle,fill=white] (0Orcrora0) at (0, 4.0) {};
\node[draw,circle,fill=white] (0Orcrora1) at (0, 4.8) {};
\node[draw,circle,fill=white] (0Orcrora2) at (0, 5.6) {};
\node[draw,circle,fill=white] (0Orcrora3) at (0, 6.4) {};
\node[draw,circle,fill=white] (1Orcrora0) at (1, 4.8) {};
\node[draw,circle,fill=white] (1Orcrora1) at (1, 5.6) {};
\node[draw,circle,fill=white] (2Orcrora0) at (2, 4.0) {};
\node[draw,circle,fill=white] (2Orcrora1) at (2, 4.8) {};
\node[draw,circle,fill=white] (2Orcrora2) at (2, 5.6) {};
\node[draw,circle,fill=white] (2Orcrora3) at (2, 6.4) {};
\node[draw,circle,fill=white] (0ao0) at (3, 2.8) {};
\node[draw,circle,fill=white] (0ao1) at (3, 3.6) {};
\node[draw,circle,fill=white] (1ao0) at (4, 2.0) {};
\node[draw,circle,fill=white] (1ao1) at (4, 2.8) {};
\node[draw,circle,fill=white] (1ao2) at (4, 3.6) {};
\node[draw,circle,fill=white] (1ao3) at (4, 4.4) {};
\draw (0Orcror0) -- (0Orcror0);
\draw (0Orcror0) -- (0Orcror0);
\draw (0Orcror0) -- (1Orcror0);
\draw (0Orcror1) -- (0Orcror1);
\draw (0Orcror1) -- (0Orcror1);
\draw (0Orcror1) -- (1Orcror0);
\draw (0Orcror2) -- (0Orcror2);
\draw (0Orcror2) -- (0Orcror2);
\draw (0Orcror2) -- (1Orcror1);
\draw (0Orcror3) -- (0Orcror3);
\draw (0Orcror3) -- (0Orcror3);
\draw (0Orcror3) -- (1Orcror1);
\draw (1Orcror0) -- (0Orcror0);
\draw (1Orcror0) -- (0Orcror1);
\draw (1Orcror0) -- (2Orcror0);
\draw (1Orcror0) -- (2Orcror1);
\draw (1Orcror0) -- (2Orcror2);
\draw (1Orcror0) -- (2Orcror3);
\draw (1Orcror1) -- (0Orcror2);
\draw (1Orcror1) -- (0Orcror3);
\draw (1Orcror1) -- (2Orcror0);
\draw (1Orcror1) -- (2Orcror1);
\draw (1Orcror1) -- (2Orcror2);
\draw (1Orcror1) -- (2Orcror3);
\draw (2Orcror0) -- (1Orcror0);
\draw (2Orcror0) -- (1Orcror1);
\draw (2Orcror0) -- (0ao0);
\draw (2Orcror1) -- (1Orcror0);
\draw (2Orcror1) -- (1Orcror1);
\draw (2Orcror1) -- (0ao0);
\draw (2Orcror2) -- (1Orcror0);
\draw (2Orcror2) -- (1Orcror1);
\draw (2Orcror2) -- (0ao1);
\draw (2Orcror3) -- (1Orcror0);
\draw (2Orcror3) -- (1Orcror1);
\draw (2Orcror3) -- (0ao1);
\draw (0Orcrora0) -- (0Orcrora0);
\draw (0Orcrora0) -- (0Orcrora0);
\draw (0Orcrora0) -- (1Orcrora0);
\draw (0Orcrora1) -- (0Orcrora1);
\draw (0Orcrora1) -- (0Orcrora1);
\draw (0Orcrora1) -- (1Orcrora0);
\draw (0Orcrora2) -- (0Orcrora2);
\draw (0Orcrora2) -- (0Orcrora2);
\draw (0Orcrora2) -- (1Orcrora1);
\draw (0Orcrora3) -- (0Orcrora3);
\draw (0Orcrora3) -- (0Orcrora3);
\draw (0Orcrora3) -- (1Orcrora1);
\draw (1Orcrora0) -- (0Orcrora0);
\draw (1Orcrora0) -- (0Orcrora1);
\draw (1Orcrora0) -- (2Orcrora0);
\draw (1Orcrora0) -- (2Orcrora1);
\draw (1Orcrora0) -- (2Orcrora2);
\draw (1Orcrora0) -- (2Orcrora3);
\draw (1Orcrora1) -- (0Orcrora2);
\draw (1Orcrora1) -- (0Orcrora3);
\draw (1Orcrora1) -- (2Orcrora0);
\draw (1Orcrora1) -- (2Orcrora1);
\draw (1Orcrora1) -- (2Orcrora2);
\draw (1Orcrora1) -- (2Orcrora3);
\draw (2Orcrora0) -- (1Orcrora0);
\draw (2Orcrora0) -- (1Orcrora1);
\draw (2Orcrora0) -- (0ao0);
\draw (2Orcrora1) -- (1Orcrora0);
\draw (2Orcrora1) -- (1Orcrora1);
\draw (2Orcrora1) -- (0ao0);
\draw (2Orcrora2) -- (1Orcrora0);
\draw (2Orcrora2) -- (1Orcrora1);
\draw (2Orcrora2) -- (0ao1);
\draw (2Orcrora3) -- (1Orcrora0);
\draw (2Orcrora3) -- (1Orcrora1);
\draw (2Orcrora3) -- (0ao1);
\draw (0ao0) -- (0ao0);
\draw (0ao0) -- (0ao0);
\draw (0ao0) -- (1ao0);
\draw (0ao0) -- (1ao1);
\draw (0ao0) -- (1ao2);
\draw (0ao0) -- (1ao3);
\draw (0ao0) -- (2Orcror0);
\draw (0ao0) -- (2Orcror1);
\draw (0ao0) -- (2Orcrora0);
\draw (0ao0) -- (2Orcrora1);
\draw (0ao1) -- (0ao1);
\draw (0ao1) -- (0ao1);
\draw (0ao1) -- (1ao0);
\draw (0ao1) -- (1ao1);
\draw (0ao1) -- (1ao2);
\draw (0ao1) -- (1ao3);
\draw (0ao1) -- (2Orcror2);
\draw (0ao1) -- (2Orcror3);
\draw (0ao1) -- (2Orcrora2);
\draw (0ao1) -- (2Orcrora3);
\draw (1ao0) -- (0ao0);
\draw (1ao0) -- (0ao1);
\draw (1ao1) -- (0ao0);
\draw (1ao1) -- (0ao1);
\draw (1ao2) -- (0ao0);
\draw (1ao2) -- (0ao1);
\draw (1ao3) -- (0ao0);
\draw (1ao3) -- (0ao1);
\end{tikzpicture}

            }
            \end{center}
            \caption{ Gate computing OR.}\label{orpicture}
            %\end{figure}
  \end{minipage}
  \hspace{0.1\textwidth}
  \begin{minipage}{0.45\textwidth}
    
            %\begin{figure}
            \begin{center}
            \scalebox{0.5}{
            
\begin{tikzpicture}
\node (Notanno1) at (-1.3, 5.7) {  \LARGE $I_1$};
\node[scale=5,color=black!30!white] (Notanno1) at (-0.4, 5.7) {  \LARGE \textbraceleft };
\node (Notanno2) at (-1.3, 0.2) {  \LARGE $I_2$};
\node[scale=5,color=black!30!white] (Notanno2) at (-0.4, 0.2) {  \LARGE \textbraceleft };
\node (Notanno3) at (5.45, 5.7) {  \LARGE $O_1$};
\node[scale=5,color=black!30!white] (Notanno3) at (4.4, 5.7) {  \LARGE \textbraceright };
\node[draw,circle,fill=white] (0NotwireInNotLow0) at (0, -1.0) {};
\node[draw,circle,fill=white] (0NotwireInNotLow1) at (0, -0.19999999999999996) {};
\node[draw,circle,fill=white] (0NotwireInNotLow2) at (0, 0.6000000000000001) {};
\node[draw,circle,fill=white] (0NotwireInNotLow3) at (0, 1.4000000000000004) {};
\node[draw,circle,fill=white] (1NotwireInNotLow0) at (1, -0.19999999999999996) {};
\node[draw,circle,fill=white] (1NotwireInNotLow1) at (1, 0.6000000000000001) {};
\node[draw,circle,fill=white] (2NotwireInNotLow0) at (2, -1.0) {};
\node[draw,circle,fill=white] (2NotwireInNotLow1) at (2, -0.19999999999999996) {};
\node[draw,circle,fill=white] (2NotwireInNotLow2) at (2, 0.6000000000000001) {};
\node[draw,circle,fill=white] (2NotwireInNotLow3) at (2, 1.4000000000000004) {};
\node[draw,circle,fill=white] (3NotwireInNotLow0) at (3, -0.19999999999999996) {};
\node[draw,circle,fill=white] (3NotwireInNotLow1) at (3, 0.6000000000000001) {};
\node[draw,circle,fill=white] (4NotwireInNotLow0) at (4, -1.0) {};
\node[draw,circle,fill=white] (4NotwireInNotLow1) at (4, -0.19999999999999996) {};
\node[draw,circle,fill=white] (4NotwireInNotLow2) at (4, 0.6000000000000001) {};
\node[draw,circle,fill=white] (4NotwireInNotLow3) at (4, 1.4000000000000004) {};
\node[draw,circle,fill=white] (0NotwireInNot0) at (0, 4.5) {};
\node[draw,circle,fill=white] (0NotwireInNot1) at (0, 5.3) {};
\node[draw,circle,fill=white] (0NotwireInNot2) at (0, 6.1) {};
\node[draw,circle,fill=white] (0NotwireInNot3) at (0, 6.9) {};
\node[draw,circle,fill=white] (1NotwireInNot0) at (1, 5.3) {};
\node[draw,circle,fill=white] (1NotwireInNot1) at (1, 6.1) {};
\node[draw,circle,fill=white] (2NotwireInNot0) at (2, 4.5) {};
\node[draw,circle,fill=white] (2NotwireInNot1) at (2, 5.3) {};
\node[draw,circle,fill=white] (2NotwireInNot2) at (2, 6.1) {};
\node[draw,circle,fill=white] (2NotwireInNot3) at (2, 6.9) {};
\node[draw,circle,fill=white] (3NotwireInNot0) at (3, 5.3) {};
\node[draw,circle,fill=white] (3NotwireInNot1) at (3, 6.1) {};
\node[draw,circle,fill=white] (0NotendWireNot0) at (4, 4.5) {};
\node[draw,circle,fill=white] (0NotendWireNot1) at (4, 5.3) {};
\node[draw,circle,fill=white] (0NotendWireNot2) at (4, 6.1) {};
\node[draw,circle,fill=white] (0NotendWireNot3) at (4, 6.9) {};
\draw (0NotwireInNotLow0) -- (0NotwireInNotLow0);
\draw (0NotwireInNotLow0) -- (0NotwireInNotLow0);
\draw (0NotwireInNotLow0) -- (1NotwireInNotLow0);
\draw (0NotwireInNotLow1) -- (0NotwireInNotLow1);
\draw (0NotwireInNotLow1) -- (0NotwireInNotLow1);
\draw (0NotwireInNotLow1) -- (1NotwireInNotLow0);
\draw (0NotwireInNotLow2) -- (0NotwireInNotLow2);
\draw (0NotwireInNotLow2) -- (0NotwireInNotLow2);
\draw (0NotwireInNotLow2) -- (1NotwireInNotLow1);
\draw (0NotwireInNotLow2) -- (3NotwireInNot0);
\draw (0NotwireInNotLow3) -- (0NotwireInNotLow3);
\draw (0NotwireInNotLow3) -- (0NotwireInNotLow3);
\draw (0NotwireInNotLow3) -- (1NotwireInNotLow1);
\draw (0NotwireInNotLow3) -- (3NotwireInNot0);
\draw (1NotwireInNotLow0) -- (0NotwireInNotLow0);
\draw (1NotwireInNotLow0) -- (0NotwireInNotLow1);
\draw (1NotwireInNotLow0) -- (2NotwireInNotLow0);
\draw (1NotwireInNotLow0) -- (2NotwireInNotLow1);
\draw (1NotwireInNotLow0) -- (2NotwireInNotLow2);
\draw (1NotwireInNotLow0) -- (2NotwireInNotLow3);
\draw (1NotwireInNotLow1) -- (0NotwireInNotLow2);
\draw (1NotwireInNotLow1) -- (0NotwireInNotLow3);
\draw (1NotwireInNotLow1) -- (2NotwireInNotLow0);
\draw (1NotwireInNotLow1) -- (2NotwireInNotLow1);
\draw (1NotwireInNotLow1) -- (2NotwireInNotLow2);
\draw (1NotwireInNotLow1) -- (2NotwireInNotLow3);
\draw (2NotwireInNotLow0) -- (1NotwireInNotLow0);
\draw (2NotwireInNotLow0) -- (1NotwireInNotLow1);
\draw (2NotwireInNotLow0) -- (3NotwireInNotLow0);
\draw (2NotwireInNotLow1) -- (1NotwireInNotLow0);
\draw (2NotwireInNotLow1) -- (1NotwireInNotLow1);
\draw (2NotwireInNotLow1) -- (3NotwireInNotLow0);
\draw (2NotwireInNotLow2) -- (1NotwireInNotLow0);
\draw (2NotwireInNotLow2) -- (1NotwireInNotLow1);
\draw (2NotwireInNotLow2) -- (3NotwireInNotLow1);
\draw (2NotwireInNotLow3) -- (1NotwireInNotLow0);
\draw (2NotwireInNotLow3) -- (1NotwireInNotLow1);
\draw (2NotwireInNotLow3) -- (3NotwireInNotLow1);
\draw (3NotwireInNotLow0) -- (2NotwireInNotLow0);
\draw (3NotwireInNotLow0) -- (2NotwireInNotLow1);
\draw (3NotwireInNotLow0) -- (4NotwireInNotLow0);
\draw (3NotwireInNotLow0) -- (4NotwireInNotLow1);
\draw (3NotwireInNotLow0) -- (4NotwireInNotLow2);
\draw (3NotwireInNotLow0) -- (4NotwireInNotLow3);
\draw (3NotwireInNotLow1) -- (2NotwireInNotLow2);
\draw (3NotwireInNotLow1) -- (2NotwireInNotLow3);
\draw (3NotwireInNotLow1) -- (4NotwireInNotLow0);
\draw (3NotwireInNotLow1) -- (4NotwireInNotLow1);
\draw (3NotwireInNotLow1) -- (4NotwireInNotLow2);
\draw (3NotwireInNotLow1) -- (4NotwireInNotLow3);
\draw (4NotwireInNotLow0) -- (3NotwireInNotLow0);
\draw (4NotwireInNotLow0) -- (3NotwireInNotLow1);
\draw (4NotwireInNotLow1) -- (3NotwireInNotLow0);
\draw (4NotwireInNotLow1) -- (3NotwireInNotLow1);
\draw (4NotwireInNotLow2) -- (3NotwireInNotLow0);
\draw (4NotwireInNotLow2) -- (3NotwireInNotLow1);
\draw (4NotwireInNotLow2) -- (3NotwireInNot1);
\draw (4NotwireInNotLow3) -- (3NotwireInNotLow0);
\draw (4NotwireInNotLow3) -- (3NotwireInNotLow1);
\draw (4NotwireInNotLow3) -- (3NotwireInNot1);
\draw (0NotwireInNot0) -- (0NotwireInNot0);
\draw (0NotwireInNot0) -- (0NotwireInNot0);
\draw (0NotwireInNot0) -- (1NotwireInNot0);
\draw (0NotwireInNot1) -- (0NotwireInNot1);
\draw (0NotwireInNot1) -- (0NotwireInNot1);
\draw (0NotwireInNot1) -- (1NotwireInNot0);
\draw (0NotwireInNot2) -- (0NotwireInNot2);
\draw (0NotwireInNot2) -- (0NotwireInNot2);
\draw (0NotwireInNot2) -- (1NotwireInNot1);
\draw (0NotwireInNot3) -- (0NotwireInNot3);
\draw (0NotwireInNot3) -- (0NotwireInNot3);
\draw (0NotwireInNot3) -- (1NotwireInNot1);
\draw (1NotwireInNot0) -- (0NotwireInNot0);
\draw (1NotwireInNot0) -- (0NotwireInNot1);
\draw (1NotwireInNot0) -- (2NotwireInNot0);
\draw (1NotwireInNot0) -- (2NotwireInNot1);
\draw (1NotwireInNot0) -- (2NotwireInNot2);
\draw (1NotwireInNot0) -- (2NotwireInNot3);
\draw (1NotwireInNot1) -- (0NotwireInNot2);
\draw (1NotwireInNot1) -- (0NotwireInNot3);
\draw (1NotwireInNot1) -- (2NotwireInNot0);
\draw (1NotwireInNot1) -- (2NotwireInNot1);
\draw (1NotwireInNot1) -- (2NotwireInNot2);
\draw (1NotwireInNot1) -- (2NotwireInNot3);
\draw (2NotwireInNot0) -- (1NotwireInNot0);
\draw (2NotwireInNot0) -- (1NotwireInNot1);
\draw (2NotwireInNot0) -- (3NotwireInNot0);
\draw (2NotwireInNot1) -- (1NotwireInNot0);
\draw (2NotwireInNot1) -- (1NotwireInNot1);
\draw (2NotwireInNot1) -- (3NotwireInNot0);
\draw (2NotwireInNot2) -- (1NotwireInNot0);
\draw (2NotwireInNot2) -- (1NotwireInNot1);
\draw (2NotwireInNot2) -- (3NotwireInNot1);
\draw (2NotwireInNot3) -- (1NotwireInNot0);
\draw (2NotwireInNot3) -- (1NotwireInNot1);
\draw (2NotwireInNot3) -- (3NotwireInNot1);
\draw (3NotwireInNot0) -- (2NotwireInNot0);
\draw (3NotwireInNot0) -- (2NotwireInNot1);
\draw (3NotwireInNot0) -- (0NotwireInNotLow2);
\draw (3NotwireInNot0) -- (0NotwireInNotLow3);
\draw (3NotwireInNot0) -- (0NotendWireNot0);
\draw (3NotwireInNot0) -- (0NotendWireNot1);
\draw (3NotwireInNot0) -- (0NotendWireNot2);
\draw (3NotwireInNot0) -- (0NotendWireNot3);
\draw (3NotwireInNot1) -- (2NotwireInNot2);
\draw (3NotwireInNot1) -- (2NotwireInNot3);
\draw (3NotwireInNot1) -- (4NotwireInNotLow2);
\draw (3NotwireInNot1) -- (4NotwireInNotLow3);
\draw (3NotwireInNot1) -- (0NotendWireNot0);
\draw (3NotwireInNot1) -- (0NotendWireNot1);
\draw (3NotwireInNot1) -- (0NotendWireNot2);
\draw (3NotwireInNot1) -- (0NotendWireNot3);
\draw (0NotendWireNot0) -- (0NotendWireNot0);
\draw (0NotendWireNot0) -- (0NotendWireNot0);
\draw (0NotendWireNot0) -- (3NotwireInNot0);
\draw (0NotendWireNot0) -- (3NotwireInNot1);
\draw (0NotendWireNot1) -- (0NotendWireNot1);
\draw (0NotendWireNot1) -- (0NotendWireNot1);
\draw (0NotendWireNot1) -- (3NotwireInNot0);
\draw (0NotendWireNot1) -- (3NotwireInNot1);
\draw (0NotendWireNot2) -- (0NotendWireNot2);
\draw (0NotendWireNot2) -- (0NotendWireNot2);
\draw (0NotendWireNot2) -- (3NotwireInNot0);
\draw (0NotendWireNot2) -- (3NotwireInNot1);
\draw (0NotendWireNot3) -- (0NotendWireNot3);
\draw (0NotendWireNot3) -- (0NotendWireNot3);
\draw (0NotendWireNot3) -- (3NotwireInNot0);
\draw (0NotendWireNot3) -- (3NotwireInNot1);
\end{tikzpicture}

            }
            \end{center}
            \caption{ Gadget computing NOT. It uses clock signal in $I_1$.}\label{notpicture}
            %\end{figure}
  \end{minipage}
\end{figure}

\subsection{proof of Lemma~\ref{lemma:everyFunction}}\label{appendix:proofEveryFunction}

\begin{proof}[Proof of Lemma~\ref{lemma:everyFunction}]
  Let $c \in \mathbb{N}$,
  and let $f: \{0,1\}^c \rightarrow \{0,1\}$
  be a function mapping any tuple whose first variable is $0$ to $0$.
  We use an induction on the dimension $c$ of the domain.
  If $c=1$,
  since $f$ satisfies $f(0) = 0$ by supposition,
  either $f$ is the identity, which is realised by a wire,
  or $f$ is the zero function,
  which is realised by simply disconnecting
  the input and output.

  Now let us suppose that $c>1$,
  and that we can realise any function
  $f:\{0,1\}^{c-1} \rightarrow \{0,1\}$ that satisfies the condition.
  In particular,
  there exist two gadgets $g_0$ and $g_1$ computing
  \[
  \begin{array}{lclc}
  f_0: & \{0,1\}^{c-1} & \rightarrow & \{0,1\},\\
  & (x_1,x_2,\ldots,x_{c-1}) & \mapsto & f(x_1,x_2,\dots,x_{c-1},0),\\
  f_1: & \{0,1\}^{c-1} & \rightarrow & \{0,1\},\\
  & (x_1,x_2,\ldots,x_{c-1}) & \mapsto & f(x_1,x_2,\dots,x_{c-1},1).
  \end{array}
  \]
  Having these values, the computation is straightforward
  as the value $f(x_1,x_2,\ldots,x_c)$ is given by the formula 
  \[
  (f_0(x_1,x_2,\ldots,x_{c-1}) \wedge \neg x_c) \vee
  (f_1(x_1,x_2,\ldots,x_{c-1}) \wedge x_c).
  \]
  We construct a gadget $g$ computing $f$ as follows.
  \subparagraph{Computing the atoms.}
  We apply two splitters to the input $I_1$ of $g$
  to obtain three copies of the received input $x_1$.
  Similarly, for every $2 \leq i \leq c$,
  we duplicate the signal $x_i$ received by
  the input $I_i$ of $g$ by applying a splitter.
  For every $1 \leq i \leq c-1$,
  we send the first copy of $x_i$ to the input $I_i$ of $g_0$,
  and the second to the input $I_i$ of $g_1$.
  We already have the signal $x_c$.
  To compute $\neg x_c$,
  we use a NOT gate $g'$,
  which gets as input $I_1$ the third copy of $x_1$,
  and as input $I_2$ the second copy of $x_c$.
  Note that if $x_1 = 0$,
  the result of $g'$ will be $0$ instead of $\neg x_c$;
  however, this is not a problem, as we know 
  that whenever $x_1 = 0$ the function has to be zero.
  \subparagraph{Combining the atoms.}
  We link the outputs of $g_0$ and $g'$ to an AND gate,
  the first copy of $x_c$ and the output of $g_1$ to another AND gate,
  and apply and OR gate to the results.
  \subparagraph{Size of the gadget.}
  Now, we show that the computation is fast and requires a reasonable number of basic gadgets.
  Each input is split in two or three,
  for this $c+1$ basic gadgets and $10$ steps are enough.
  To compute $x_c$ and $\neg x_c$ we use a single basic gadget and at most $10$ steps.
  By the induction hypothesis, both $g_0$ and $g_1$
  can be constructed using $2^{c+10}$ basic gadgets,
  and have an evaluation time of at most $10c$.
  Connecting $f_1,f_0,x_c$ and $\neg x_c$ needs only three basic gadgets and at most $10$ steps.
  This means that $g$ is composed of
  $\calO(2^c)$ basic gadgets,
  and its evaluation time is $\calO(c)$.
\end{proof}

\subsection{Regularization of the graph}\label{appendix:regularization}
\begin{theorem}
  The construction of Theorem~\ref{theorem:overpopulation} also holds for regular graphs with degree $10$.
\end{theorem}
\begin{proof}
Every vertex in the graph
$G= (V,E)$ presented in the proof of Lemma \ref{lemma:simulating_t_m} has constant
degree, and we can use this fact to make the constructed
graph regular. Amongst our basic gadget, the highest degree
is $10$: both the splitter and the storage unit reach this bound
with two vertices (Figure \ref{splitterpicture} and \ref{supicture}).
Note that, according to
the overpopulation rule $\overpopulation{2}{1}$, any vertex that has at most
one live neighbor will stays dead all the time. Therefore, we
can add leaves (vertices with degree one) so that every vertex
in our construction has degree exactly ten, and we know for
sure that the leaves will never become alive. Finally, we need
to add nine edges to every leaf, so that each leaf also has
degree ten. Note that we added $10|V|-2|E|$ leaves, which is
an even number. Therefore we can take any 9-regular graph
over the leaves, for instance the union of nine pairings.
\end{proof}
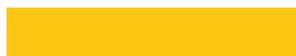
\end{document}